%% file: main.tex
\documentclass[10pt, sigconf]{acmart}

\renewcommand\footnotetextcopyrightpermission[1]{} 
\usepackage{booktabs} 
\usepackage{graphicx}
\usepackage{subfigure}
\usepackage{amsmath}
\usepackage[noend]{algpseudocode}
\usepackage{varwidth}
\usepackage{url}
\usepackage{comment}
\usepackage{multirow}
\usepackage{color}

\begin{document}

\title{Tiny Buffer TCP for Data Center Networks}

\author{Lixia Xiong, Nan Li}
\affiliation{%
 \institution{Huawei Datacom Research Department}
 \city{Beijing}
 \state{China}}
\email{{xionglixia, lee.linan}@huawei.com}

\author{Haoyu Song}
\affiliation{%
  \institution{Futurewei Technologies}
  \city{Santa Clara}
  \state{USA}}
\email{haoyu.song@futurewei.com}

\begin{abstract}

A low and stable buffer occupancy is critical to achieve high throughput, low packet drop rate, low latency, and low jitter for data center networks. It also allows switch chips to support higher port density,  larger lookup tables, or richer functions. Tiny Buffer TCP creatively uses the common RED-based ECN with two novel congestion-window adjustment schemes to significantly reduce the required buffer size. Aiming to eliminate the residual packets in the bottleneck queue, Queue Canceling Decrease amortizes the ideal window reduction to the same number of flows so as to minimize the impact to active flows. In order to keep the buffer occupancy low and stable, Reduced Additive Increase recovers the flow window at a slower pace than normal. We implemented TBTCP in Linux kernel and conducted $ns2$-based simulations and real network-based tests. Our results show that compared to DCTCP,  TBTCP reduces the switch buffer requirement by more than 80\%, increases the bandwidth utilization by up to 15\%, improves the FCT performance by up to  39\%, and achieve a 71\% better RTT fairness index.

\end{abstract}

\maketitle

\section{Introduction}

Data Center Network (DCN) is critical in supporting cloud applications. Numerous works have been published in recent years to improve the performance of DCN by proposing new switch design~\cite{conga, fastpass}, new transport protocol~\cite{dctcp, bbr, timely, phost, dx}, and the combination of both~\cite{pfabric, d3, dcqcn, NDP}.  

While the commodity switch hardware is slow to pick up new algorithms and new architectures, a more appealing and effective option is to \textbf{only modify the end-host transport protocol software} for better performance, given a data center is basically an autonomous system. Indeed, DCTCP~\cite{dctcp} has been included in official Linux kernel release and adopted in many data centers as the default transport protocol. It also serves as the foundation for some newer algorithms~\cite{dcqcn, pias}. In a multi-tenant data center where the guest's transport protocol cannot be modified, software innovation at edge is still preferred over installing new switch hardware. It is shown that new transport algorithms can be emulated by host hypervisors without modifying guest OSes~\cite{acdctcp, vcc}.

Buffer is the main facility in switches to work in conjunction with end-host congestion control algorithms.  It helps switches retain bandwidth utilization and absorb traffic burst. The buffer occupancy, which is reflected in queue depth, is a direct signal for network congestion status. While some chip vendors boast a bigger buffer at the cost of other resources, we argue that this can do more harm than good. 

First, buffer poses a serious design challenge for switches. Because the off-chip memory is slow and demands a large amounts of I/Os, the state-of-art switch chips have to use on-chip SRAM as packet buffer.  A chip can afford no more than a few tens of megabytes of memory due to the foot print of SRAM cells. In a typical switch chip with balanced resource distribution, roughly 1/3 of the chip area is dedicated to memory, 1/3 to I/Os, and the remaining 1/3 to logics. The memory is further partitioned to support lookup tables/queues and packet buffer. As the switch chips become more flexible and programmable~\cite{rmt}, the requirement for more logics and lookup tables escalates. Meanwhile, the throughput and port density of switch chips continue to increase. A 12.8Tbps switch chip, which includes 128x 100Gbps ports, is on the horizon. This trend implies that I/Os will claim more chip area and the average buffer size per port will shrink. If we are able to reduce the buffer size, we can make switch chips scale to higher throughput and port density, or reserve more resources for richer functions and higher flexibilities. 

Second, data center applications have strict latency and throughput requirements. Simply using large buffer to absorb bursts is lazy and ineffective. An unattended large buffer can lead to poor performance (i.e., long queuing delay, high jitter, and excessive buffer bloat). For example, a 1MB filled buffer introduces an extra 200us of queuing delay at a 40Gbps port, while the RTT of a light-loaded DCN is less than 250us~\cite{dctcp}. In an extreme case, the 20MB fully shared buffer in a switch can be occupied by a single congested port and the worst-case queuing delay will be up to 4ms.

The delay-sensitive applications in data centers expect lower RTT (e.g., <100us) and smaller RTT jitter (e.g., <50us). The only hope to achieve this performance target is to keep the switch buffer occupancy near zero and stable regardless of the network condition. We believe the end-host congestion control still has the potential to help achieve this target. Tiny Buffer TCP (TBTCP) is therefore designed to meet the modern DCN's need: lowering the chip buffer requirement while improving the flow transport performance. 



\input{motivation}

\input{algorithm}

\input{evaluation}

\input{implementation}

\input{test}
\input{summary}
\input{related}

\section{Conclusion}

In spite of the prosperity of DCN research, TCP remains the most popular data center transport protocol and it still has room to improve.
Small and stable queue occupancy in switches is essential for flow performance and switch efficiency. 
The core design principle of TBTCP is to enable finer control of TCP congestion window, which helps avoid drastic queue oscillation under dynamic traffic conditions. 
TBTCP is instantly deployable. It outperforms DCTCP in jitter, latency, throughput, FCT, and RTT fairness. Moreover, TBTCP's small buffer requirement allows data center switches to continue scaling for higher throughput, port density, and flexibility.  We believe TBTCP is a capable contender of DCTCP in data centers. 

In our future work, in addition to exploring the hardware improvements, we will study TBTCP's applicability in wide area networks, where router buffer is also facing the scalability issue~\cite{buffersizing}. We will also try to apply the underlying QCD and RAI mechanisms in the RDMA protocol stack and explore the possibility to improve the performance of RDMA-based DCNs.   


\bibliographystyle{acm}
\bibliography{sigproc} 

\end{document}

%% file: motivation.tex
\section{DCTCP's Performance Issue}~\label{motivation}

A DCN spans only a few hops in a limited geographic distance, so the packet propagation and processing delay is small. The queuing delay is the main contributor to packet transport latency, and the RTT is directly proportional to the queue depth. Many Active Queue Management (AQM) algorithms have been developed to drop or mark packets based on different queue conditions~\cite{pie, codel, cedm, mqecn}. Since these algorithms are designed to interact with normal TCP, their performance cannot meet the data center requirements.


For data center applications, Flow Completion Time (FCT) is the most important criterion in measuring the data center network performance~\cite{fct}. The relationship between FCT and queue depth is subtle. For small flows, only a few RTTs are needed to finish the packet transmission, which prevent the flow window from growing large. In this case, FCT is sensitive to the queuing delay, making it desirable to keep the queue depth shallow. If the queue is deep, the queuing delay is prolonged and the probability of buffer overflow is increased. Packet drops for small flows are especially catastrophic because the RTO can be many times longer than the normal flow FCT. For large flows, while it is important to avoid buffer overflow, it is equally important to avoid the buffer underflow in order to maintain the full throughput.  

In summary, the bottleneck queue depth should stay as low as possible, but not completely empty if active flows are present. If this is achieved under any traffic condition, the FCT for both large and small flows can profit. In data centers, DCTCP presents a better FCT performance than TCP~\cite{dctcp}, which means that DCTCP's buffer control is better than TCP. However, DCTCP is still far from ideal. Modern data center applications (e.g., big data analytics) raise the performance requirement bar higher than what DCTCP can offer, which begs for further improvements. Next, we examine the underlying relationship between DCTCP's performance and the queue depth/buffer size in switches.

DCTCP uses a single queue threshold $k$ for ECN mark. If the queue depth exceeds $k$ packets when enqueuing a packet, the packet is marked. The ECN mark is piggybacked by a returning ACK packet from the receiver to the sender.  The sender counts the ratio of marked packets in a flow and adjusts the flow's congestion window accordingly.   Specifically, the estimated ratio $\alpha$ is updated every RTT as follows, in which $g$ is the weight parameter and $F$ is the actual counted mark ratio in the last window:

\begin{equation}
\alpha \leftarrow (1-g)\alpha+gF 
\end{equation}

In response to an ECN mark, the sender updates its congestion window as:

\begin{equation}\label{eq_dctcp}
\texttt{cwnd} \leftarrow \texttt{cwnd}(1-\alpha/2)
\end{equation}

Theoretical analysis shows that the maximum queue depth (as well as the buffer size) should be at least $k+n$ (where $n$ is the number of concurrent flows) to avoid buffer overflow. Meanwhile, $k$ should be at least $(C*\texttt{RTT})/7=\texttt{BDP}/7$ (where $C$ is the bottleneck link bandwidth) to avoid buffer underflow. Although $k$ is smaller than the BDP, it is proportional to it and can be of large value. 

Consider a 40Gbps link and a 250us RTT. In this case, $k$ is 178KB or equal to 118 MTU-sized packets. Even worse, the actual queue depth can be anywhere between $0$ to $k+n$. In reality, $n$ can be large (e.g., 1000 current flows are not uncommon in data center workloads). Therefore, the queuing delay and jitter can be large, and the switch needs a large buffer.
In this example, the filled queue will introduce more than 300us extra delay. Figure~\ref{fig_dctb} illustrates the queue depth dynamics of DCTCP (as well as TBTCP for comparison). 

\begin{figure}[!t]
\centering
\includegraphics[width=.5\columnwidth]{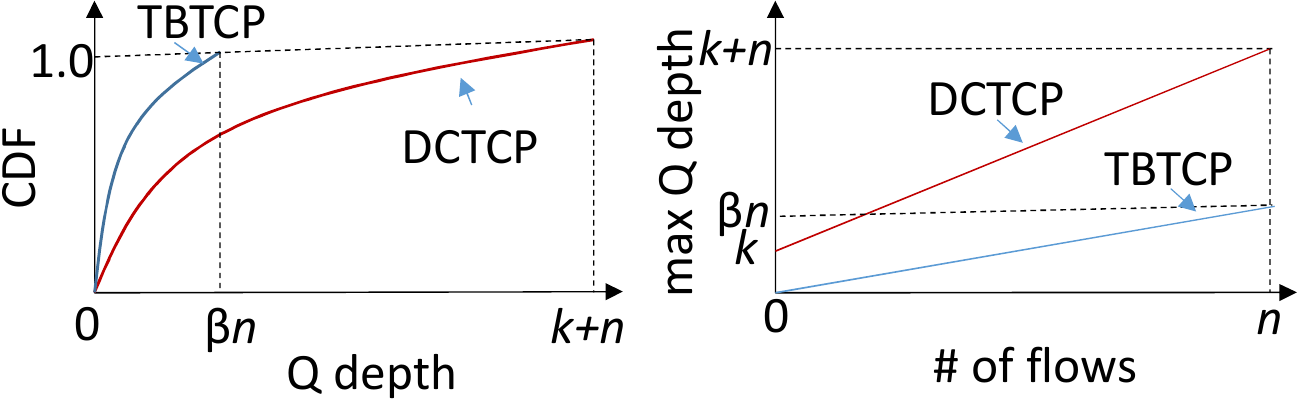}
\caption{Queue Depth for DCTCP and TBTCP}
\label{fig_dctb}
\end{figure}

This wide queue-oscillating range leads to poor FCT for small flows. The large buffer required is also a cost burden to switch. The buffer requirement of DCTCP is coupled with the BDP and the number of flows $n$, which seriously limits DCTCP's scalability. An interesting question is: \emph{can we lower the buffer size requirement so better FCT performance and lower switch cost can be realized simultaneously?} We believe DCTCP's issue is rooted in its coarse and imprecise congestion window adjustments. Next we study the optimal window adjustment strategy and show how a fine, agile, and precise window control can answer this question. 

\section{Fine Control the Buffer Size}\label{decouple}

\subsection{Achieve Optimal BDP}\label{qcd}

According to Gail and Kleinrok's analysis, the optimal network operating condition happens at the inflection point in Figure~\ref{fig_bdp}, where the in-flight data equals the Bandwidth-Delay Product (BDP) and the RTT is minimum~\cite{bbr}. At this point, the bottleneck link bandwidth is fully utilized but the bottleneck queue is empty. The current TCP algorithms, including DCTCP, tend to produce a larger amount of in-flight data, which deviates from the optimal point and causes long queue in switches. While BBR~\cite{bbr} uses direct network measurements to approach the optimal point, TBTCP strives to work as close to the optimal point as possible with a simpler scheme. 
  
\begin{figure}[!t]
\centering
\includegraphics[width=0.6\columnwidth]{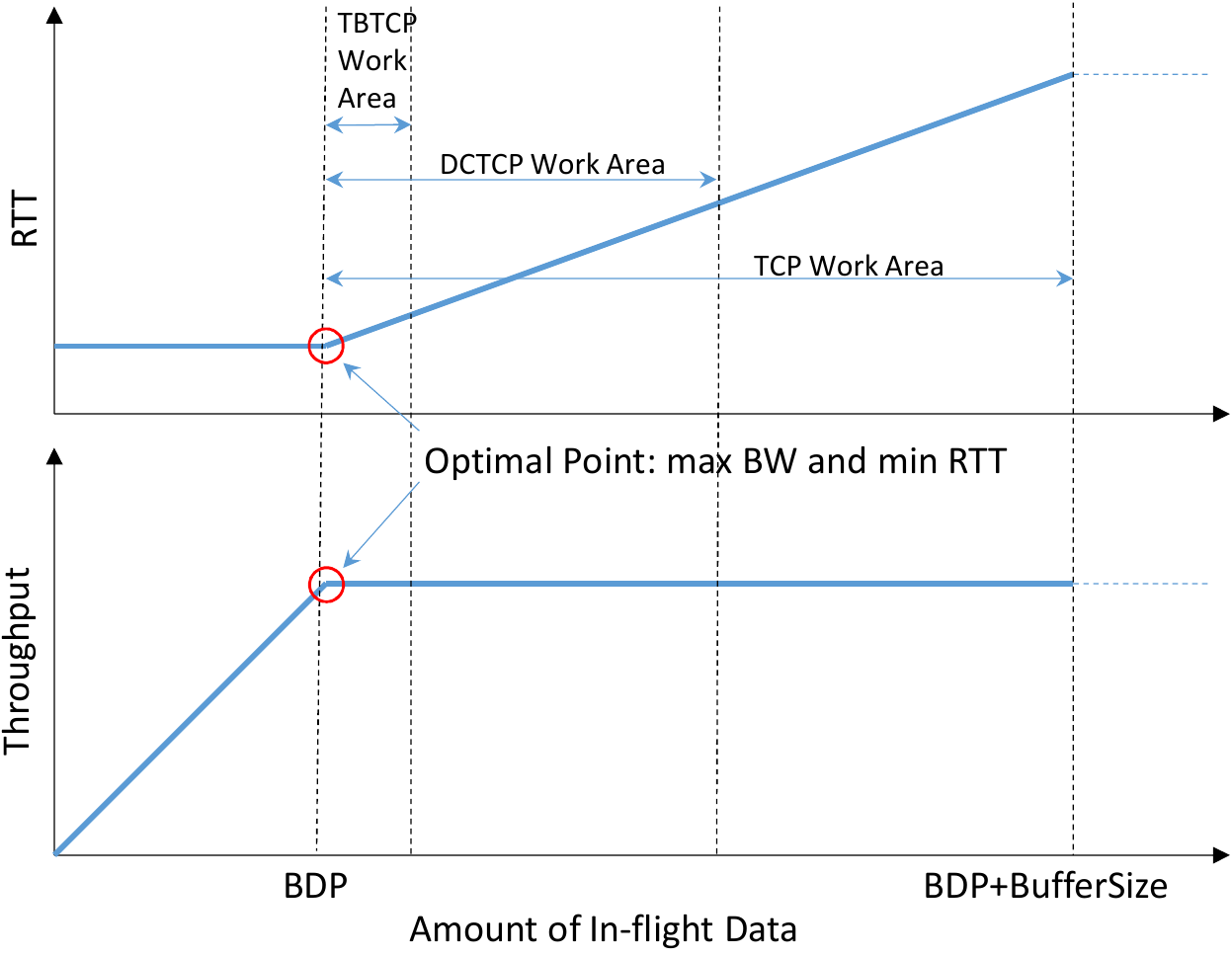}
\caption{TBTCP Towards Optimal Work Point}
\label{fig_bdp}
\end{figure}

In DCTCP, the congestion control algorithm makes the queue depth vary between $\texttt{max}\{0, k-n\}$ and $k+n$ with the center of $k$. In contrast, TBTCP aims to reduce the steady-state queue depth to the range of $\{0, \beta n\}$, where $\beta << 1$.

For convenience, we consider all the flows sharing the bottleneck link as one aggregated flow $F$, which window is the sum of the flow windows. The congestion avoidance algorithm tries to dynamically reduce $F$'s congestion window so as to eliminate the buffer bloat and maintain 100\% bottleneck bandwidth utilization. Assume the congestion window size of $F$ is $W_{max}$ and $W_{min}$ before and after a window adjustment. The following equation should hold true:

\begin{equation}\label{eq_c}
C = \frac{W_{max}}{\texttt{RTT}_{max}} = \frac{W_{min}}{\texttt{RTT}_{min}}
\end{equation}

In Equation~\ref{eq_c}, $\texttt{RTT}_{max} = 2T_p + Q/C$ and $\texttt{RTT}_{min}= 2T_p$, where $T_p$ is the one-way propagation delay and $Q$ is the bottleneck queue depth. Hence,
 
\begin{equation}
\frac{W_{max}}{2T_p + Q/C} = \frac{W_{min}}{2T_p}
\end{equation}

We can deduce that the right amount of window adjustment for $F$ in one RTT, $\delta$, is:

\begin{equation}\label{eq_delta}
\delta = W_{max} -W_{min} = \frac{QW_{max}}{2T_pC+Q} 
\end{equation}

We draw two conclusions from Equation~\ref{eq_delta}. First, the smaller the queue depth $Q$ is, the smaller the window adjustment $\delta$ needs to be made. For example, if $Q=2T_pC$, then $\delta = W_{max}/2$, which agrees with the classical TCP's behavior; if $Q=2T_pC/7$, as suggested in DCTCP, then $\delta = W_{max}/8$. This implies that DCTCP needs a much smaller window adjustment. As shown in Equation~\ref{eq_dctcp}, a factor of $\alpha$ is applied for the window adjustment. However, since $\alpha$ cannot track the queue depth well (recall that  $\alpha$ only tracks the ECN mark ratio for a single queue threshold $k$), the adjustment is inaccurate, leading to deep queue and high queue oscillation. 

Second, if the queue depth $Q$ is small enough, then $\delta \approx Q$.  This is because when $Q \ll 2T_pC$, we can get $W_{max} \approx 2T_pC = \texttt{BDP}$. Therefore, when the queue depth is small, we should try to reduce the window size by the exact amount that equals the queue depth in each RTT.  This makes the network work closer to the optimal point and avoid buffer underflow as shown in Figure~\ref{fig_bdp}.

Based on these observations, we develop the first key feature of TBTCP: Queue Canceling Decrease (QCD). That is, \emph{in each RTT, we reduce the overall flow window by exactly the size of the bottleneck queue depth.} Moreover, we distribute the window reductions to as many flows as possible so each affected flow only needs to reduce its window by one MSS. We defer the discussion on how QCD is realized to Section~\ref{algorithm}.

\subsection{Smooth Queue Oscillation}\label{rai}

Even though QCD can make precise window reduction, the following window recovery may be too fast, which can nullify the reduction effects and introduce drastic queue oscillations.  

For example, the amplitude of DCTCP's  queue oscillation goes up to $n$. We refer to this as the coupling issue between the buffer size and the concurrent flow count. The fundamental reason is the additive increase of congestion windows: in every RTT, a flow's window increases one Maximum Segment Size (MSS) after the slow start phase. The additive increase works in wide area networks, but appears to be excessive for DCN due to its short RTT. For example, if RTT is 100us, increasing a flow's \texttt{cwnd} by one is equivalent to adding $1500B*8b/B/100us = 120Mbps$ network bandwidth consumption. For 100 concurrent flows, the instant bandwidth surge is up to 12Gbps. As a result, the bottleneck queue becomes volatile and the flow performance suffers.  

To avoid the bandwidth surge, we would like to reduce the window recovery speed. Specifically, we take $m>1$ RTTs to increase a flow's \texttt{cwnd} by one MSS. This is equivalent to reducing the amount of window size increase in each RTT to $1/m$. 
Hence we develop the second key feature of TBTCP: Reduced Additive Increase (RAI). That is, \emph{in each RTT, we virtually increase a flow's \texttt{cwnd} by $\beta = \frac{1}{m} < 1$} (In real implementation, the window increase by 1 in $m$ RTTs). Note that the seemingly ``slow'' increase is offset by the fact that the reduction is small in the first place.

We have the following theorem:

\begin{theorem}
Due to QCD and RAI, the bottleneck queue depth will be stable at $\beta n$ in one RTT and the overall \texttt{cwnd} approximates to the bottleneck BDP. 
\end{theorem}


\begin{proof}
Assume that at time $t$ the queue depth is $Q$ and the total window size of the $n$ flows is $W$. At time $t+$RTT, $W$ is reduced by $Q$ due to QCD; meanwhile, $W$ is increased by $\beta n$ due to RAI. Hence, the new overall window size $W' \leftarrow W+\beta n-Q$. Since $W'$ is increased by $\delta = \beta n-Q$, the new queue depth $Q'$ will become $Q + \delta = \beta n$. 

Following the same deduction process, at time $t+2$RTT, $W'$ (as well as $Q'$) is reduced by $\beta n$ due to QCD and increased by $\beta n$ due to RAI, so $W'$ (as well as $Q'$) is no longer changed (although the distribution of $W'$ among the $n$ flows is changing).  

That is, the queue depth $Q'$ is stable at $\beta n$ and the overall \texttt{cwnd} size $W'$ is stable at $W+\beta n-Q$ after the first RTT, if the flow number $n$ remains unchanged.

Based on the analysis in Section~\ref{qcd},  $W_{max} \approx $ BDP due to the small $Q'$.
\end{proof}

Since the queue depth can grow to at most $\beta n$, TBTCP's queue oscillation amplitude is more than $1/\beta$ times smaller than DCTCP's.

%% file: algorithm.tex
\section{TBTCP Algorithm}~\label{algorithm}

Like DCTCP, the TBTCP algorithm contains three components which work in switch, receiver, and sender, respectively.

\subsection{RED-based Marking on Switches} 

Random Early Detection (RED) is a commodity active queue management scheme in modern switches. It monitors the instantaneous queue size and marks packets with the CE codepoint 
based on statistical probabilities. Specifically, it sets two queue depth thresholds, \texttt{t\_min} and \texttt{t\_max}. The packets are randomly marked and the marking probability linearly increases from 0 to \texttt{P\_max} as the queue depth grows from \texttt{t\_min} to \texttt{t\_max}. 

TBTCP takes advantage of the RED scheme. Recall that QCD requires the overall flow window to be reduced by the size of the bottleneck queue depth in each RTT. To realize QCD, we want to mark exactly $Q$ packets when the queue depth is $Q$.  


From Equation~\ref{eq_delta}, we can derive that $\frac{Q}{2T_pC+Q}$ is the ideal packet marking probability to ensure the proper window adjustment. That is: 

\begin{equation}\label{eq_pq}
p(Q) = \frac{Q}{\texttt{BDP}+Q} 
\end{equation}

This is because in one RTT, $BDP+Q$ packets are sent through the bottleneck buffer. With the marking probability $p(Q)$, exact $Q$ packets will be marked. This simplified analysis assumes that in one RTT the queue depth remains unchanged and the synchronized window adjustments happen at each RTT boundary. In reality, the queue depth constantly changes. However, the number of marked packets in an RTT is bounded by the largest queue depth during this RTT.

The ideal marking probability graph appears to be a curve, but RED can only provide a polyline graph. We need to fit the two graphs so that RED can approximate the ideal marking probability. As shown in Figure~\ref{fig_red}, we can set $\texttt{t\_min} =  0$ and $\texttt{P\_max} = 1$, so the actual marking probability $p(Q) = \frac{Q}{\texttt{t\_max}}$.  By choosing the proper \texttt{t\_max}, we can get a good approximation of the marking probability. 

\begin{figure}[!t]
\centering
\includegraphics[width=0.6\columnwidth]{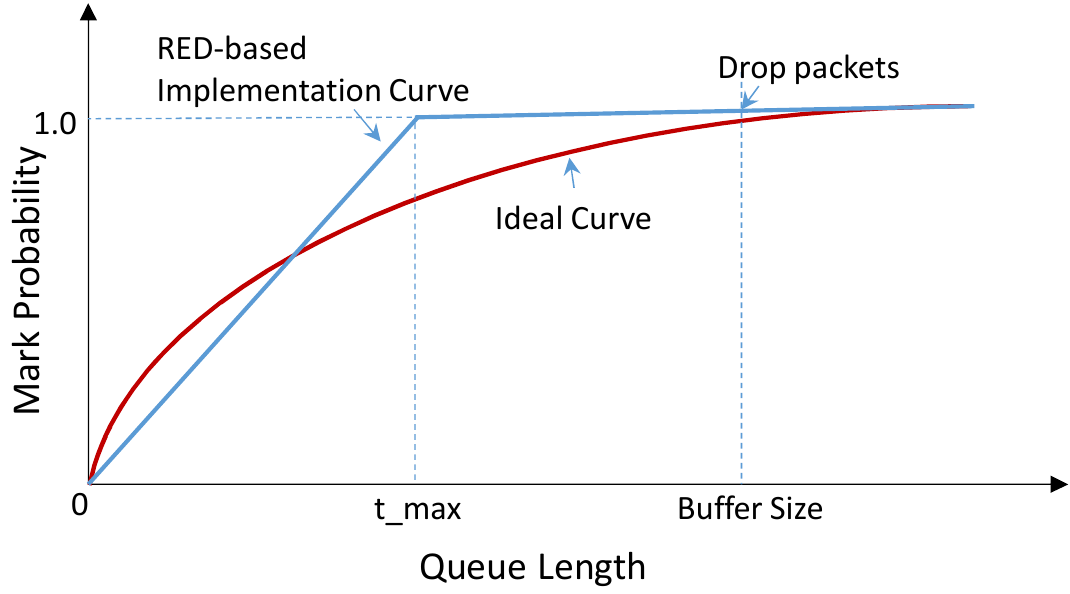}
\caption{RED ECN Probability Curve for QCD}
\label{fig_red}
\end{figure}

\subsection{ECN Echo from Receivers} 

Once a marked packet is received at a receiver, the receiver needs to echo the information back to the sender by setting the ECN-Echo flag in an ACK packet. In case the delayed-ACK feature is preferred, we modify the delayed-ACK behavior to ensure the timely feedback. Whenever a marked packet is received, an ACK with the ECN-Echo flag set for the accumulated packets so far is immediately sent; otherwise, a normal delayed ACK is sent for every a few fixed number of packets. 
 

\subsection{Congestion Control at Senders}

Senders host TBTCP's core congestion avoidance mechanisms, QCD and RAI.  Other features, such as slow start and retransmission on packet loss, remain the same as DCTCP. 

\noindent\textbf{Queue Canceling Decrease:}
The RED-based ECN guarantees that in each RTT, exactly $Q$ packets are marked when the bottleneck queue depth is $Q$, so every time a flow sender receives an ECN, it can simply reduce its \texttt{cwnd} by one MSS. 

The total window reduction is amortized to up to $Q$ randomly selected flows. The selected flows are referred as victim flows. Clearly, since large flows send more packets in an RTT, they have a higher probability to become victim than small ones. This is in line with the idea that small flows should be less disturbed.


\noindent\textbf{Reduced Additive Increase:}
After the slow start stage (i.e., when the flow receives its first ECN), a flow's \texttt{cwnd} increases $\beta$ every RTT.  This means that after every $1/\beta$ RTTs, the flow can actually increase its \texttt{cwnd} by one \texttt{MSS}. For example, if $\beta$ is set to 0.1, then each flow's \texttt{cwnd} will be increased by one for every 10 RTTs. By doing so, the bottleneck queue depth will become stable at $0.1n$. 


\subsection{Analysis on Steady State Window} 

To analyze the steady state behavior of TBTCP, we assume that $n$ long-lived and synchronized flows, with the same RTT, share a bottleneck link of capacity $C$. In the steady state, the overall window size $W = \texttt{BDP}+Q$ and each flow's window size $w$ is:

\begin{equation}\label{eq_w}
w = \frac{W}{n} = \frac{\texttt{BDP}+Q}{n}
\end{equation} 

In each RTT, a flow's window is increased by $\beta$ due to RAI.  Now we consider QCD's effect. 

Since each packet is marked with a probability of $p(Q)$, if it starts with window size $w$, the expected new window size $w' = (w-1) p(Q) + w (1-p(Q)) = w-p(Q)$. That is, each sending packet will reduce the flow's window size by $p(Q)$. During one RTT, the flow sends $w$ packets. Therefore, due to QCD, the flow's window is reduced by $w p(Q)$.

To maintain the steady state, we need to make the window increase and decrease equal in each RTT. That is, $\beta = w p(Q)$. 

Substitute  Equation~\ref{eq_w} and \ref{eq_pq} into the above equation and we get:

\begin{equation}\label{eq_qmax}
\beta n = Q
\end{equation}

This exactly reproduces the key observation of the TBTCP algorithm: the steady state bottleneck queue depth is stable at $\beta n$. It suggests that, (1) when $\beta$ is fixed, the queue depth becomes shallower as the number of  flows is reduced; and (2) when the number of flows is fixed, a smaller $\beta$ leads to a shallower queue. 


%% file: evaluation.tex
\section{Algorithm Evaluation}~\label{eval}

To validate TBTCP, we begin by simulating the algorithm using a $ns2$ simulator. The network topology is a simple dumbbell in which $n$ flows share a single bottleneck link with the bandwidth of $C$. We monitor the bottleneck queue depth $Q$. All the TCP flows last 10 to 30 seconds. We enable the switch's  RED-based ECN feature and use Equation~\ref{eq_pq} to calculate the ECN marking probability.

\subsection{Queue Depth and Buffer Occupancy}

\textbf{TBTCP Queue Depth on $\beta$:} First, we examine the parameter $\beta$'s effect on TBTCP's performance. In this test, 100 flows share a 40Gbps link and RTT is set to 160us. Figure~\ref{fig_beta} shows that the stable queue depth is close to $\beta n$ with small swing. Even when $\beta$ is as small as 0.1, queue underflow does not occur, due to small queue oscillation. In the following simulations, we set $\beta$ to 0.1 to maximize the buffer reduction potential. 

\begin{figure}[!tbh]
\centering
\includegraphics[width=.5\columnwidth]{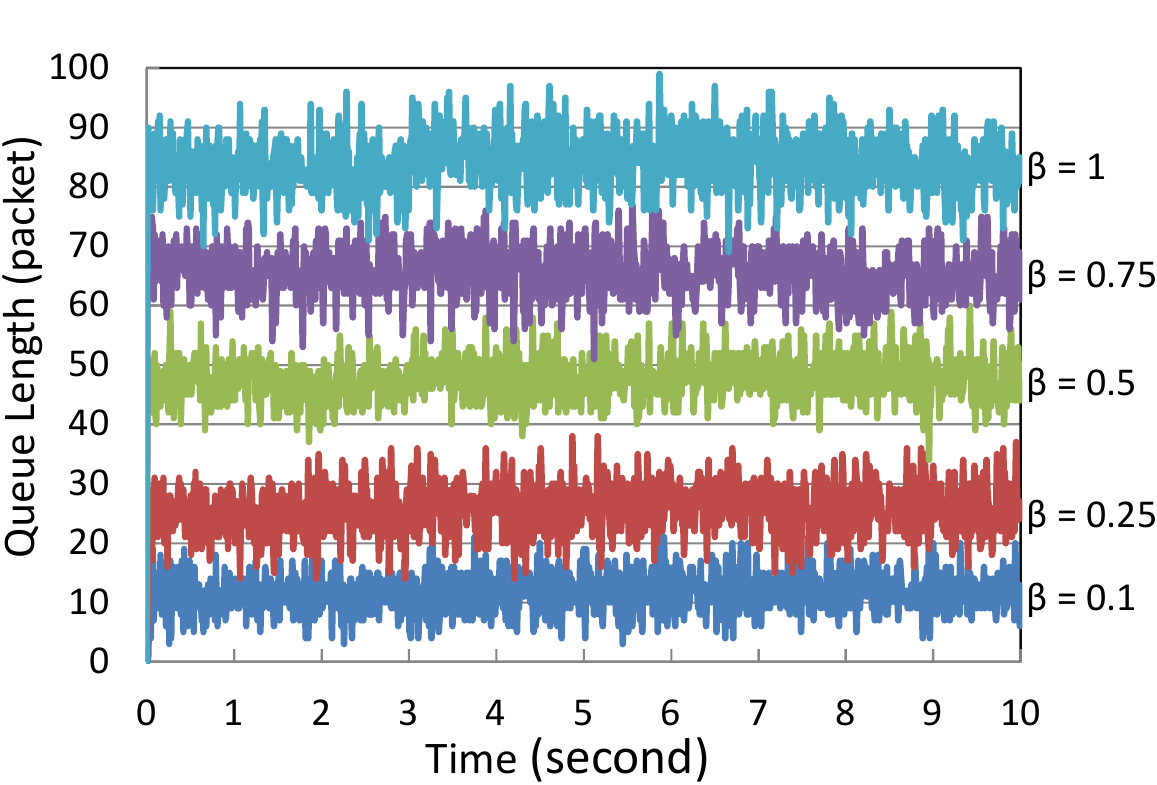}
\caption{TBTCP Queue Depth with Different $\beta$}
\label{fig_beta}
\end{figure}

\noindent\textbf{TBTCP with Ideal ECN Marking Probability:} When the ECN marking probability follows the ideal curve as shown in Figure~\ref{fig_red}, the TBTCP queue depth for different configuration is shown in Figure~\ref{fig_tbtcp_queue}.

\begin{figure}[!tbh]
\centering
\includegraphics[width=.5\columnwidth]{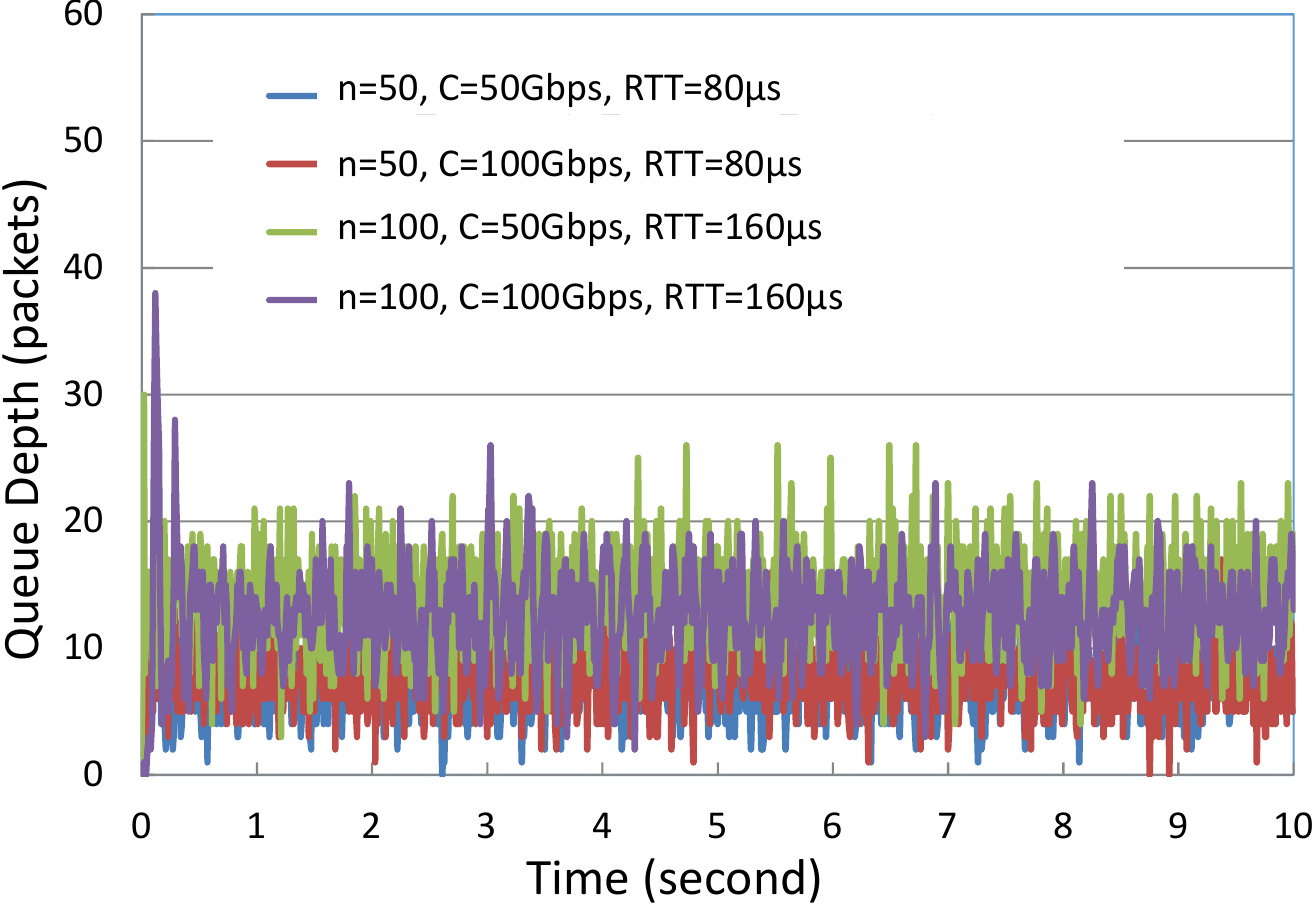}
\caption{TBTCP Queue Depth}
\label{fig_tbtcp_queue}
\end{figure}

In these simulations, the average queue depth is around 5 for 50 flows and 10 for 100 flows. Besides, the queue oscillation range is small --- about two times of the average queue depth. The queue underflow is rarely observed, which means the bottleneck bandwidth is fully utilized.

\noindent\textbf{RAI and Buffer Size:} To demonstrate that RAI can reduce the buffer size requirement even if  the number of flows is large, we implement a DCTCP variant for which the window recovery algorithm is replaced by RAI (a.k.a. DCTCP with RAI). The queue depth distribution, for different algorithms (i.e., the original DCTCP, DCTCP with RAI, and TBTCP) and under different configurations, are shown in Figure~\ref{fig_buffer}.

\begin{figure}[!tbh]
\centering
\includegraphics[width=.8\columnwidth]{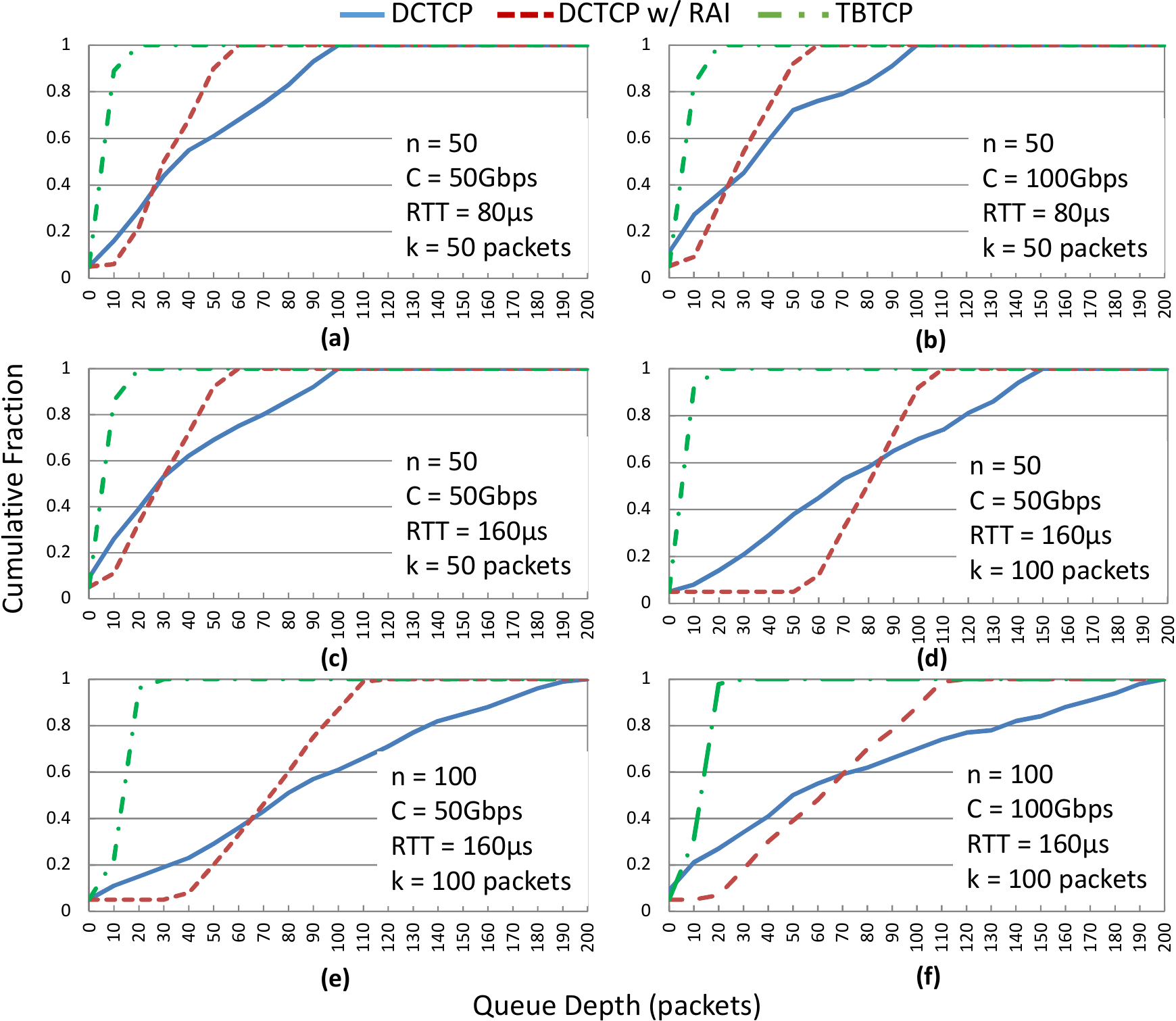}
\caption{Buffer Size CDF for Different Configurations}
\label{fig_buffer}
\end{figure}

The simulations conform to DCTCP's analysis on buffer size requirement (i.e., at least $n+k$ in order to avoid packet drop). Furthermore, the queue depth of DCTCP oscillates between 0 and $n+k$ with nearly uniform distribution, which implies a volatile queuing delay as well as RTT. On the other hand, when DCTCP is implemented with RAI, the buffer size requirement is no longer related to $n$ but just slightly higher than $k$. This confirms that RAI helps decouple the buffer size requirements from the number of concurrent flows, which is critical to implement a small-buffer switch. 

Regardless of various configurations, TBTCP constantly requires a small buffer with about 20 packets. The small buffer implies small jitter and latency for flow transport. We can also see that when the other conditions are constant, the different bottleneck bandwidth and the flow RTT have little effect on the queue depth for all three algorithms.  

RAI effectively narrows the queue oscillating range and reduces the jitter of queuing delay. Since the minimum queue depth tends to be greater than 0, RAI also helps lower the value of $k$ for DCTCP without leading to buffer underflow and throughput loss. We show the results of another simulation in Figure~\ref{fig_reducek}. In this case, we assume that 100 flows share a 40Gbps link and the RTT is 160us. The theoretical $k$ value for DCTCP is 76 packets, whereas the simulation shows that $k$ must be increased to 90 to avoid buffer underflow. In contrast, if RAI is enabled, $k$ can be reduced to 38 packets without losing any link throughput. 

\begin{figure}[!tbh]
\centering
\includegraphics[width=.5\columnwidth]{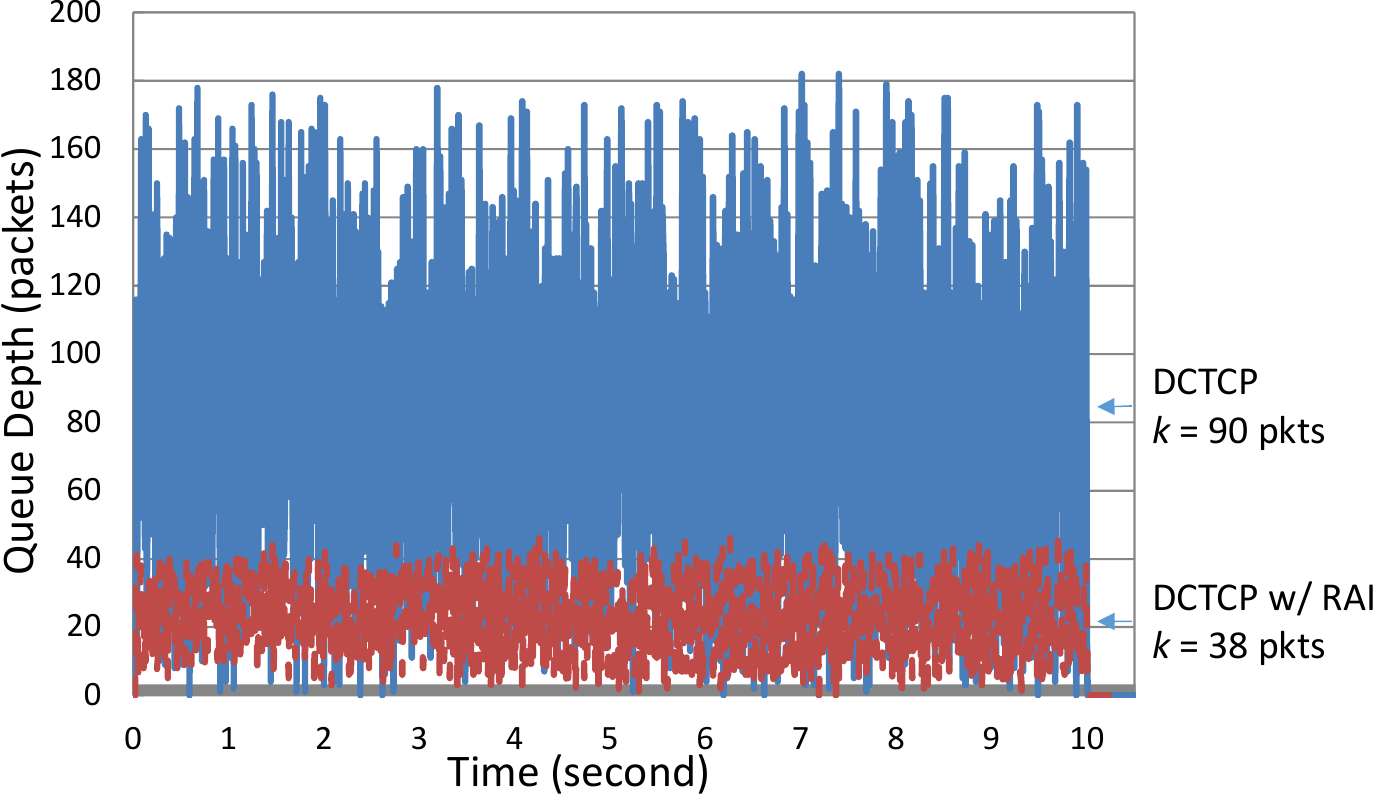}
\caption{DCTCP with RAI Can Reduces $k$}
\label{fig_reducek}
\end{figure}

\subsection{Fairness Between TBTCP Flows}

To test TBTCP's fairness, we start a new flow every 10 seconds with each flow lasting 40 seconds. Four flows are started in total. We repeat the test on a 1Gbps bottleneck link and a 10Gbps bottleneck link. The  \texttt{CWND} of each flow as a function of time are shown in Figure~\ref{fig_fair}.

\begin{figure}[!tbh]
\centering
\includegraphics[width=.8\columnwidth]{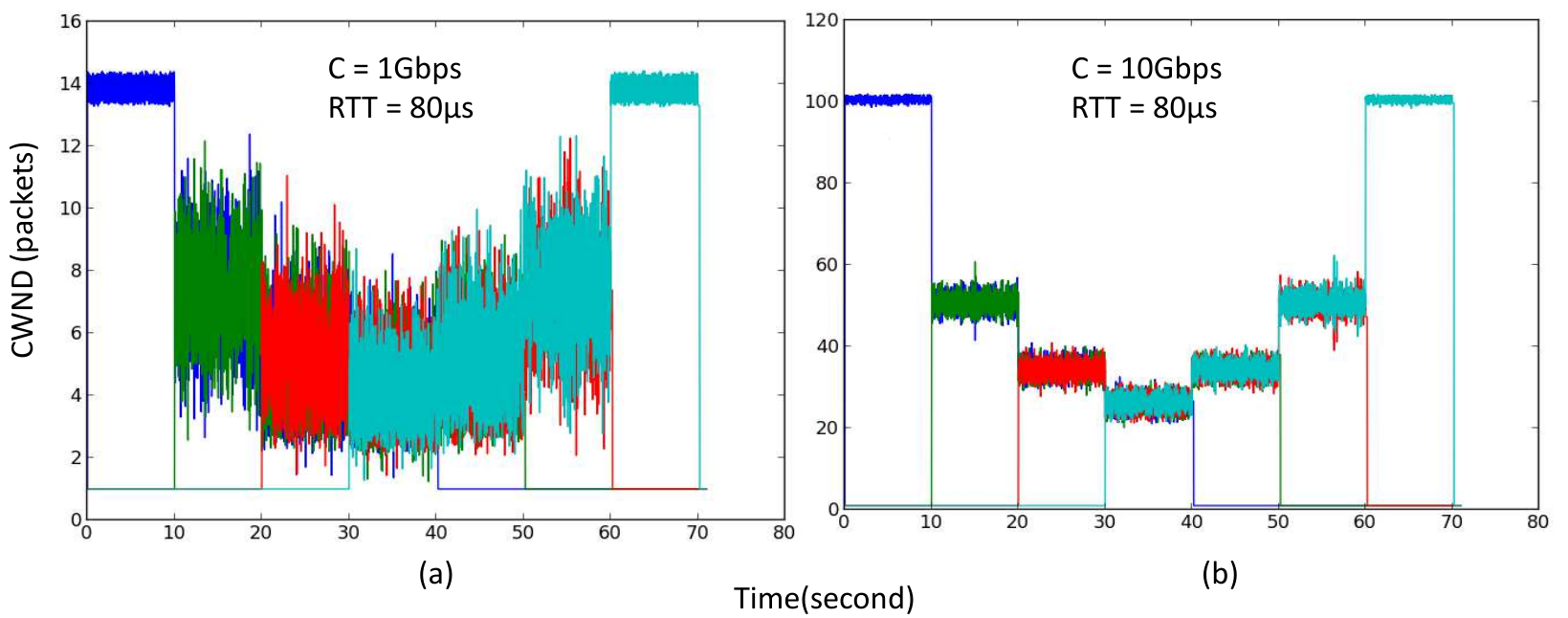}
\caption{TBTCP Fairness}
\label{fig_fair}
\end{figure}

In both cases, the bottleneck link reaches 100\% throughput. The \texttt{CWND} curve shows that each flow gains its fair share of bandwidth. When there is more than one concurrent flow, the oscillation range of the \texttt{CWND} value is larger than the case of a single flow. However, the range becomes insensitive to the number of flows past one. This is due to the complex interaction between flows: more flows increase the queue oscillation, so the ECN marking probability is higher and in turn the \texttt{CWND} adjustments are more frequent.   

\subsection{Convergence Speed}

Since RAI recovers the \texttt{CWND} at a slower pace than the normal TCP and DCTCP, the flow convergence time for gaining fair share of bandwidth may seem concerning. Figure~\ref{fig_converge} shows the simulation results for TBTCP, DCTCP, and DCTCP with RAI enabled. Each case has two flows on a 10Gbps with 80us RTT: one flow starts at 0 second and stops after 20 seconds, and the other flow starts at the 10th second and stops after 20 seconds. We examine the \texttt{CWND} transition of both flows at the 10th second and the 20th second. TBTCP takes 60ms to converge at both check points; DCTCP takes 40ms to converge at the first check point and 10ms at the second; DCTCP with RAI takes 250ms to converge at the first check point and 120ms at the second. 

\begin{figure}[!thb]
\centering
\includegraphics[width=.7\columnwidth]{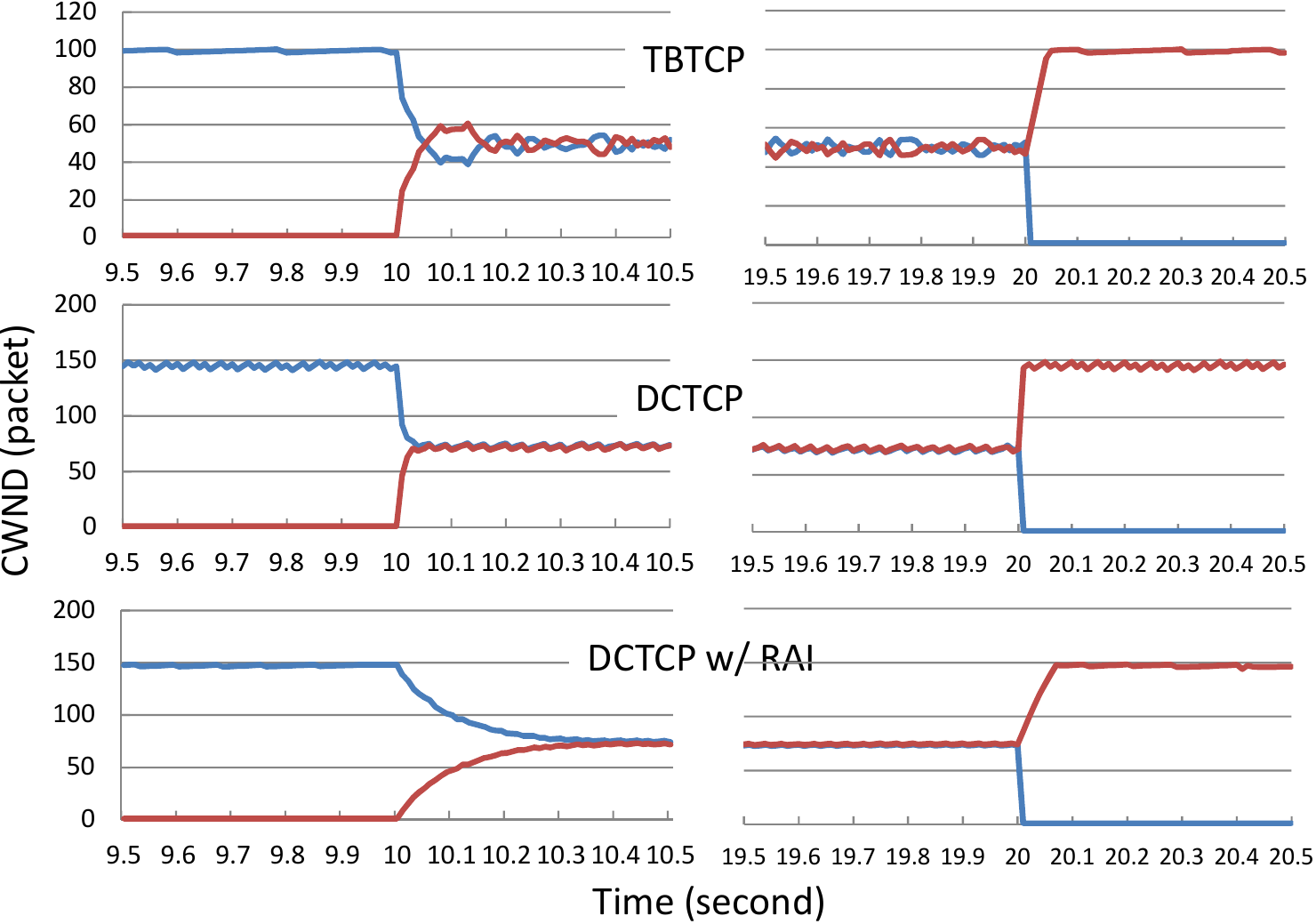}
\caption{Convergence Speed Comparison}
\label{fig_converge}
\end{figure}

DCTCP with RAI is the worst in terms of convergence time, but TBTCP is just slightly worse than DCTCP. When more concurrent flows are involved, the convergence time difference will be further reduced for the following reasons: First, the slow-start phase of a flow still uses the normal additive increase mechanism --- RAI only taking effect at the congestion-avoidance phase --- so a new flow can quickly converge to its fair share of bandwidth at the same speed that some existing flows give up their bandwidth share. In TBTCP, flows reduce speed based on QCD. The aggregated reduction is comparable  to the window reduction in DCTCP. Second, although RAI slows down the window recovery, QCD amortizes the window reduction over multiple flows, with each flow only reducing its \texttt{CWND} by one each time during congestion control. The slow recovery is accompanied by a moderate reduction. 

\noindent\textbf{Convergence Speed on Different $\beta$:} To examine how $\beta$ can affect the convergence speed, we repeat the above experiment for TBTCP with 
different $\beta$ values. As shown in Figure~\ref{fig_betaconverge}, a larger $\beta$ value can reduce the bandwidth convergence time. Clearly, we can use a larger $\beta$ value as a trade-off to accelerate the convergence speed at the cost of a larger buffer size.

We also note that when $\beta$ is equal to or greater than 0.5, TBTCP loses its micro adjustment efficacy. Experiments show that new flow can never gain its fair share of bandwidth. This is because a large $\beta$ value leads to a large stable queue depth and marking probability, which in turn restrains the new flow's window increase potential.   

\begin{figure}[!thb]
\centering
\includegraphics[width=.5\columnwidth]{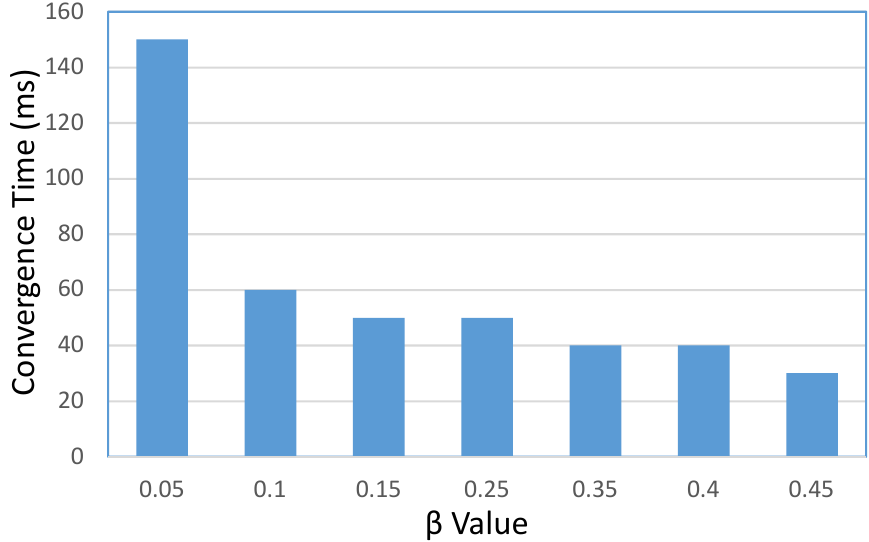}
\caption{Convergence Speed for Different $\beta$}
\label{fig_betaconverge}
\end{figure}

\subsection{RTT Fairness}

To test the flow behavior of TBTCP in contrast to DCTCP when experiencing different RTTs, we assume four long-lived flows from different sources compete for a 40Gbps bottleneck link. Each flow's maximum bandwidth is 20Gbps. The RTT of these flows is  configured to be 20us, 60us, 100us, and 140us, respectively (assuming the bottleneck queue is empty).  For DCTCP, we set $k$ to 60KB based on the theoretical formula $k = \texttt{RTT}*C/7$, in which RTT is taken the average value of 80us. For TBTCP, the ECN marking probability is $Q/(\texttt{BDP}+Q)$, in which BDP is calculated with the RTT setting to the minimum, average, and maximum values, respectively (i.e., 50us, 80us, and 110us). The experiments last 10 seconds each.

\begin{figure*}[!tb]
\centering
\includegraphics[width=1.7\columnwidth]{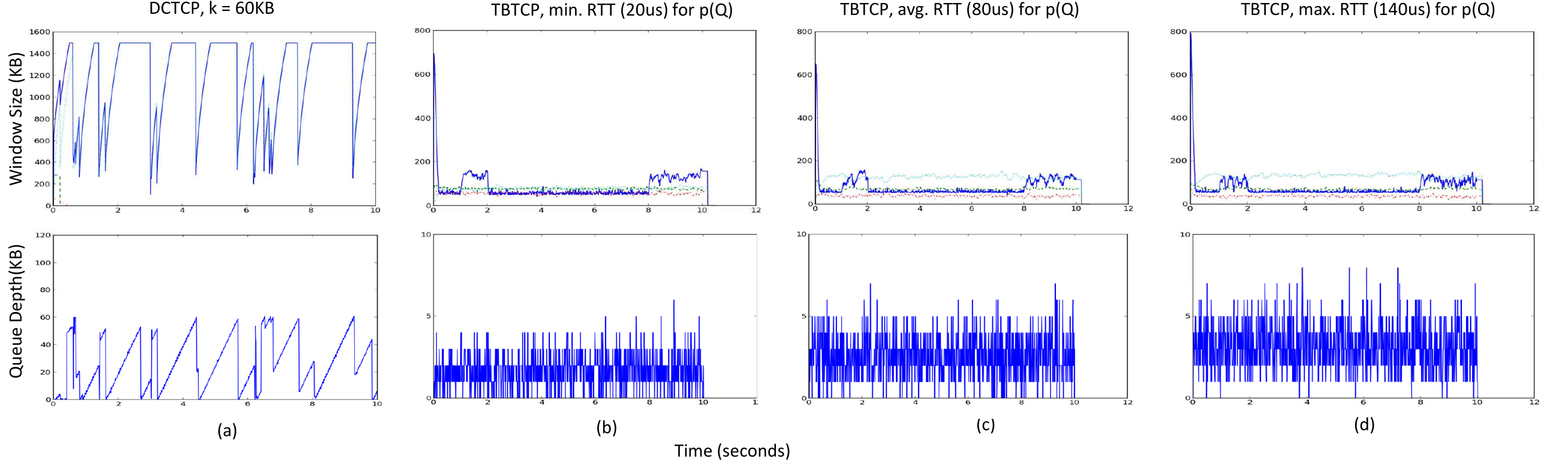}
\caption{RTT Fairness Comparison}
\label{fig_rttfair}
\end{figure*}

As shown in Figure~\ref{fig_rttfair}(a),  the RTT fairness of DCTCP in this case is poor. The Jain's fairness index is only 0.51. The two flows with smaller RTT seize almost all the bandwidth, and the other two flows with longer RTT are permanently starved. Meanwhile, the bottleneck queue depth also fluctuates violently. When we use smaller or larger $k$ values, the fairness starts getting better, at the cost of bandwidth loss or memory consumption.  

On the other hand, TBTCP presents a much better RTT fairness. The RTT value used to calculate the ECN marking probability $p(Q)$ can affect the RTT fairness. In the above three configurations, their Jain's fairness index are 0.98, 0.87, and 0.85, respectively. It appears that when the smallest RTT value is used for $p(Q)$, the fairness is the best and the average bottleneck queue depth is the shallowest. However, in this case, we can occasionally observe the queue underflow hurt the throughput. As a trade-off, we consider the average RTT a good balance between fairness and queuing performance. 

The simulation results also show a significant mismatch of the total window size between DCTCP and TBTCP. This is due to the buffer bloat introduced by $k$ in DCTCP (recall that $W=\texttt{BDP}+Q$). However, DCTCP's large window does not translate into higher throughput. TBTCP achieves the full throughput with the near optimal total window size and the minimum latency.

Note that the RTT imbalance in DCN is a real issue, given the queuing delay contribution has been significantly reduced. Assume the processing and transport delay of one switch is 15us. In a 3-tier DCN, the topology-induced RTT can range from the minimum value of 30us to the maximum value of 150us. TBTCP is preferred to DCTCP for a better balanced performance in such networks. 

\subsection{Analysis}

\noindent\textbf{RTT Fairness.} Flows may experience different RTT due to the path length variance. TCP flow's throughput is observed to be inversely proportional to RTT because flows with smaller RTT increase their window size faster. DCTCP also exhibits a negative bias against flows with longer RTT~\cite{dctcpana}. This is a serious concern for DCN where the longest RTT can be multiple times of the shortest RTT.

However, the RTT fairness is not a big issue for TBTCP. It has been shown that RED can help improve RTT fairness~\cite{rttfair}. On the other hand, due to RAI, the flows with smaller RTT increase their window sizes just slightly faster than the other flows. They will also receive ECNs faster so their window size reduction is faster. The net effect is that the RTT unfairness is well controlled. 

\noindent\textbf{Small Flow Fairness.}
Since TBTCP's window recovery speed is slow, the small flows, if their windows are reduced, would suffer more than large flows. This is mitigated by several factors: 1) Small flows have smaller probability to be chosen for window reduction; 2) RAI only takes effect after the slow start stage; 3) QCD only takes effect when the flow's window is greater than 2; and 4) the minimized RTT in TBTCP can compensate the window size loss (i.e., halving the RTT is equivalent to doubling the flow throughput). 

Summing up all these factors, TBTCP shows a surprisingly better small-flow FCT performance than DCTCP.  

\noindent\textbf{Mixed Flow Fairness.} It should be obvious that if TBTCP flows coexist with other type of flows such as TCP and DCTCP, the performance of TBTCP flows can be negatively affected. This means TBTCP can only be exclusively deployed as the only congestion control algorithm in a data center.  

\noindent\textbf{Incast.} Incast~\cite{incast} is not much of  a concern for TBTCP. Incast happens when a large number small flows compete for the same queue. However, the incast bursts build up in a few RTTs during which TBTCP can generate enough ECNs to curb enough number of flows in time.

%% file: implementation.tex
\section{Implementation}~\label{imp}

TBTCP requires a few simple modifications to DCTCP's congestion control algorithm. It also takes advantages of the RED-based ECN feature, which is generally available in commodity switches. 

\subsection{Host TCP Stack Modification}

Implementing QCD and RAI in existing TCP protocol code in host OS is straightforward. The modification is made based on the DCTCP implementation in linux-source-4.4.0. Only 10 lines of code were added or modified for the sender, and 30 for the receiver. The modifications are distributed in four functions: \emph{tcp\_cong\_avoid\_ai()}, \emph{dctcp\_ssthresh()}, \emph{tcp\_enter\_recovery()}, and \emph{dctcp\_cwnd\_event()}.  
 
\subsection{Switch Configuration}

\noindent\textbf{Shaping the ECN Probability Curve.} Figure~\ref{fig_red} shows the ideal ECN marking probability curve for QCD. However, in real implementation, we need to consider two issues.

First, many factors may increase TCP burstiness~\cite{dctcp}. For example, the NIC's Large Send Offload (LSO) function would send a big chunk of data with multiple MTU-sized packets in a burst. Since such bursts are transient and do not reflect an overall sending window boost, it is best to avoid triggering QCD. The measure we take is similar to that in DCTCP. We shift the ideal curve to the right and make the curve start from a threshold $l$. In our current implementation, we set $l$ to 138KB. This is because our NIC enables a 128KB segment size, which implies a burst size of 86 packets. Likewise, the DCTCP's $k$ should be adjusted to be at least 120 packets or 180KB to avoid triggering excessive ECNs. When the queue depth is below $l$, no packet is ECN marked. The new ECN marking probability derived from Equation~\ref{eq_pq} is therefore: 

\begin{equation}\label{eq_npq}
p(Q) = \frac{Q-l}{\texttt{BDP}-l+Q} 
\end{equation}

Second, as queue depth increases, the ECN marking probability will eventually approach 1. When the probability is high enough, there is a possibility that the sending window is overly reduced to cause buffer underflow. To counter this, the implementation of TCP and DCTCP allows the flow's sending window to be reduced by a maximum of 50\% in the case of packet drop or ECN. Similarly, our implementation limits the ECN marking probability to a maximum of 0.5. This means that no more than 50\% of packets can be marked in a RTT, which is equivalent to reducing each flow's \texttt{cwnd} by a maximum of 50\% in a RTT when the network is heavily congested. Since $p(Q)$ is a monotone increasing function, when $Q \ge \texttt{BDP}+l$, the ECN marking probability $p(Q)$ is saturated at $0.5$.

Hence, the corrected ECN marking probability is:

\[
p(Q) = \left\{ \begin{array}{ll}
				0                                          & Q \leq l \\
				\displaystyle \frac{Q-l}{\texttt{BDP}-l+Q} & l < Q \leq \texttt{BDQ} +l \\
				0.5                                        & Q > \texttt{BDQ}+l
				\end{array}
		\right.
\] 

\noindent\textbf{Fitting the Curve to Switch ECN Configuration.} For convenience and simplicity, the RED function in commercial switch chips is actually implemented as a step function that approximates the theoretical linear function~\cite{bcmguide}. Specifically, the range between $\texttt{t\_min}$ and $\texttt{t\_max}$ is divided into eight equal-sized intervals. Hence, the step size $s = (\texttt{t\_max}-\texttt{t\_min})/8$ and the $i$-th step interval $A_i = (\texttt{t\_min}+i*s, \texttt{t\_min}+(i+1)*s]$. At interval $i$, the step value $\alpha_i = (i+0.5)/12.5\%$. The actual RED function is:

\[
f(Q) = \left\{ \begin{array}{ll}
					0                               & Q \leq \texttt{t\_min} \\
					\displaystyle \sum_{i=0}^{7} \alpha_i \chi_{A_i}(Q)   & \texttt{t\_min} < Q \leq \texttt{t\_max} \\
					1                               & Q > \texttt{t\_max}
				\end{array}
		\right. 
\]	

In the above equation, $\chi_{A}$ is the indicator function of interval $A$. Our goal is to find the most fitting $f(Q)$ for $p(Q)$ and set the parameters $\texttt{t\_min}$, $\texttt{t\_max}$, and $\texttt{P\_max}$ accordingly. For simplicity, we set $l = \texttt{t\_min}$. We also know that $\texttt{t\_max} \gg \texttt{BDP}+l$. Once the queue depth reaches $\texttt{t\_max}$, we follow the RED rule to mark all packets in order to avoid buffer overflow. 

With the above considerations, the fitting problem can be translated into an optimization one:

\[
\displaystyle \min_{\texttt{P\_max}} \texttt{err} = \int_{\texttt{t\_min}}^{\texttt{t\_max}} (p(Q)-f(Q))^2\,dQ
\]

We depict $p(Q)$ and $f(Q)$ in Figure~\ref{fig_fitting}. The optimization goal is to find the optimal value of $\texttt{P\_max}$ that can minimize the square error between the two curves. 

\begin{figure}[!t]
\centering
\includegraphics[width=.6\columnwidth]{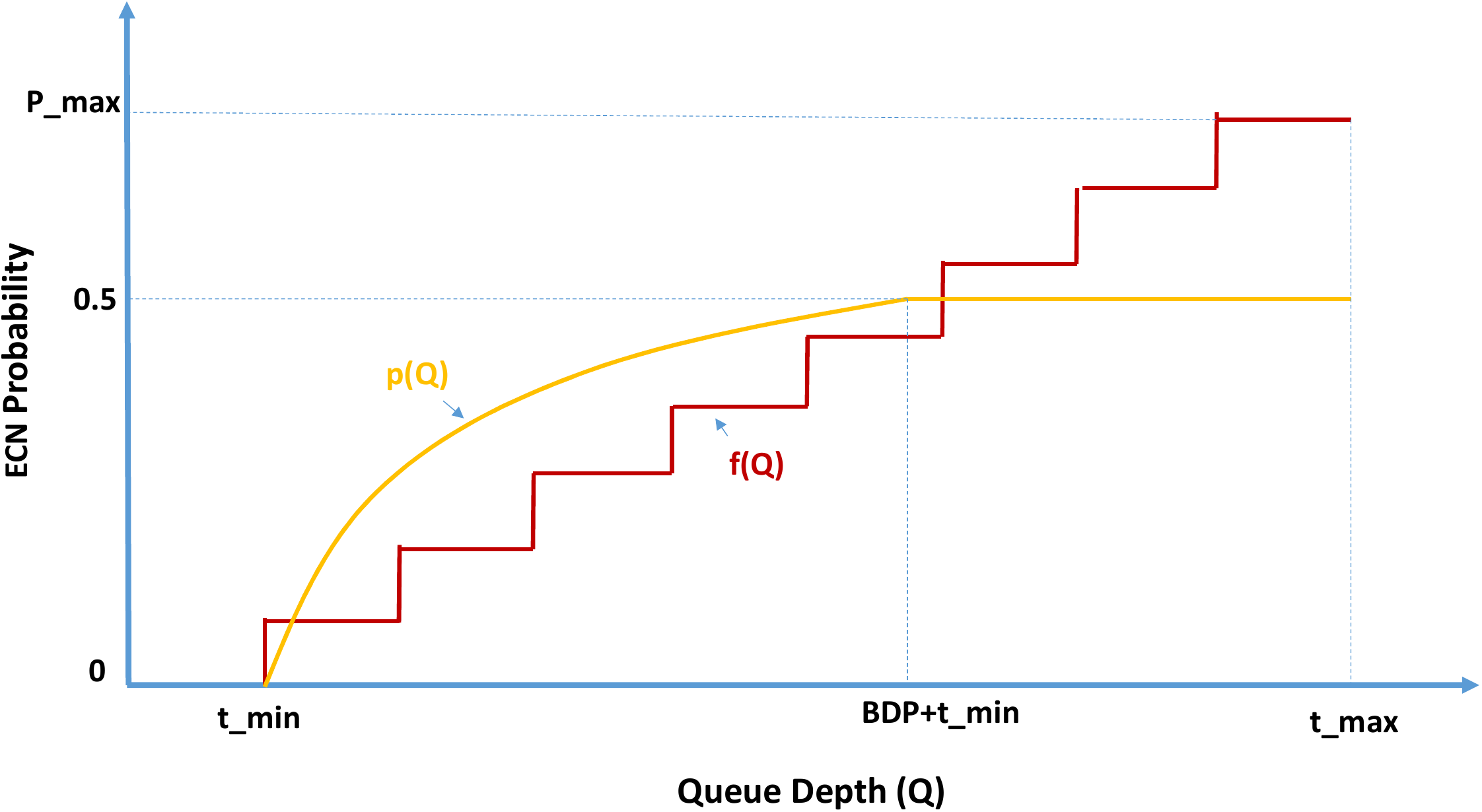}
\caption{Fitting the Switch ECN Configuration}
\label{fig_fitting}
\end{figure}

Numerical analysis shows that when $\texttt{t\_min}$, $\texttt{t\_max}$, and \texttt{BDP} are set to 138KB, 550KB and 180KB, respectively, the proper $\texttt{P\_max}$ is 0.7 and the resulting square error is 6.66.

\noindent\textbf{Scaling the Probability Curve.} The above fitting result is not satisfactory. If we use this configuration directly, the results will significantly diverge from the ideal. However, if we scale the curve of $p(Q)$ by a factor of $1/r$ (i.e., the probability at each point of $Q$ becomes $p(Q)/r$), the curve appears more flat and linear. Thus, we can find the best $\texttt{P\_max}$ for $f(Q)$ to fit $p(Q)$ with a reduced fitting error. 

Table~\ref{tab_fit} shows the optimal fitting results for different values of $r$, with the other parameters being the same as above.

{\tiny
\begin{table}[!ht]
\centering
\caption{Scaling Factor and Fitting Error}\label{tab_fit} 
\begin{tabular}{|c|l|l|} \hline
 $r$     &  $\texttt{P\_max}$  & \texttt{err}   \\ \hline\hline
1   &    0.7        &    6.66               \\ \hline
2   &     0.35     &     1.67                 \\ \hline
3   &     0.25     &    0.76                    \\ \hline
4   &      0.2     &      0.49                  \\ \hline
\end{tabular}
\end{table}
}

The new curve reduces the ECN marking probability by a factor of $r$, which is equivalent to reducing the number of ECN-marked packets by a factor of $r$. To compensate for this, a victim flow needs to reduce its congestion window by $r$ instead of $1$ whenever an ECN is received. Similar to DCTCP, TBTCP does not allow the window to be reduced by more than 50\% for each ECN. 
So, the actual QCD window adjustment for a flow is: $\texttt{CWND} \leftarrow max\{\texttt{CWND}-r,  \texttt{CWND}/2\}$.

Although a larger $r$ can help reduce the fitting error, the larger window size reduction causes greater performance impact on victim flows. It is necessary to evaluate the different configurations of $r$ and come up with the best tradeoff.

%% file: test.tex
\section{Test Results}~\label{test}

To test the performance of the TBTCP implementation, we set up a physical network which includes a few Ethernet switches and servers. Each switch has 24x 10GE ports and a Broadcom Trident switch chip. The chip has 9MB shared buffer in total and each port can use up to 1MB. Each ESX server has two CPUs with each CPU containing six 2.4GHz cores. Each server is equipped with 8x 10GE NIC, so we can start up to eight VMs each with a dedicated NIC. The VM guest OS is Ubuntu 16.04. We use iPerf~\cite{iperf} for performance measurement. In each guest, we  run either an iPerf client or an iPerf server. Each iPerf client can initiate multiple TCP flows towards an iPerf server. 

The test bed configuration is shown in Figure~\ref{fig_testbed}. Multiple clients send multiple TCP flows to a single server. The bottleneck link bandwidth is 10Gbps and the measured end-to-end RTT is 150us on average, so the \texttt{BDP} is about 180KB and the theoretical $k$ value for DCTCP is about 26KB. 

In the tests, the throughput is the sum of the measured throughputs for all the flows at the iPerf server. Due to the link layer overhead, the theoretical maximum throughput can only reach 95\% of the bottleneck link bandwidth.  Meanwhile, the iPerf software overhead also causes some throughput loss, especially when the number of flows is large. Since the real-time buffer occupancy in the switches is invisible, the amount of throughput loss caused by buffer underflow cannot be determined. In the following tests, we use DCTCP as the benchmark and compare it with TBTCP. 

\begin{figure}[!th]
\centering
\includegraphics[width=.5\columnwidth]{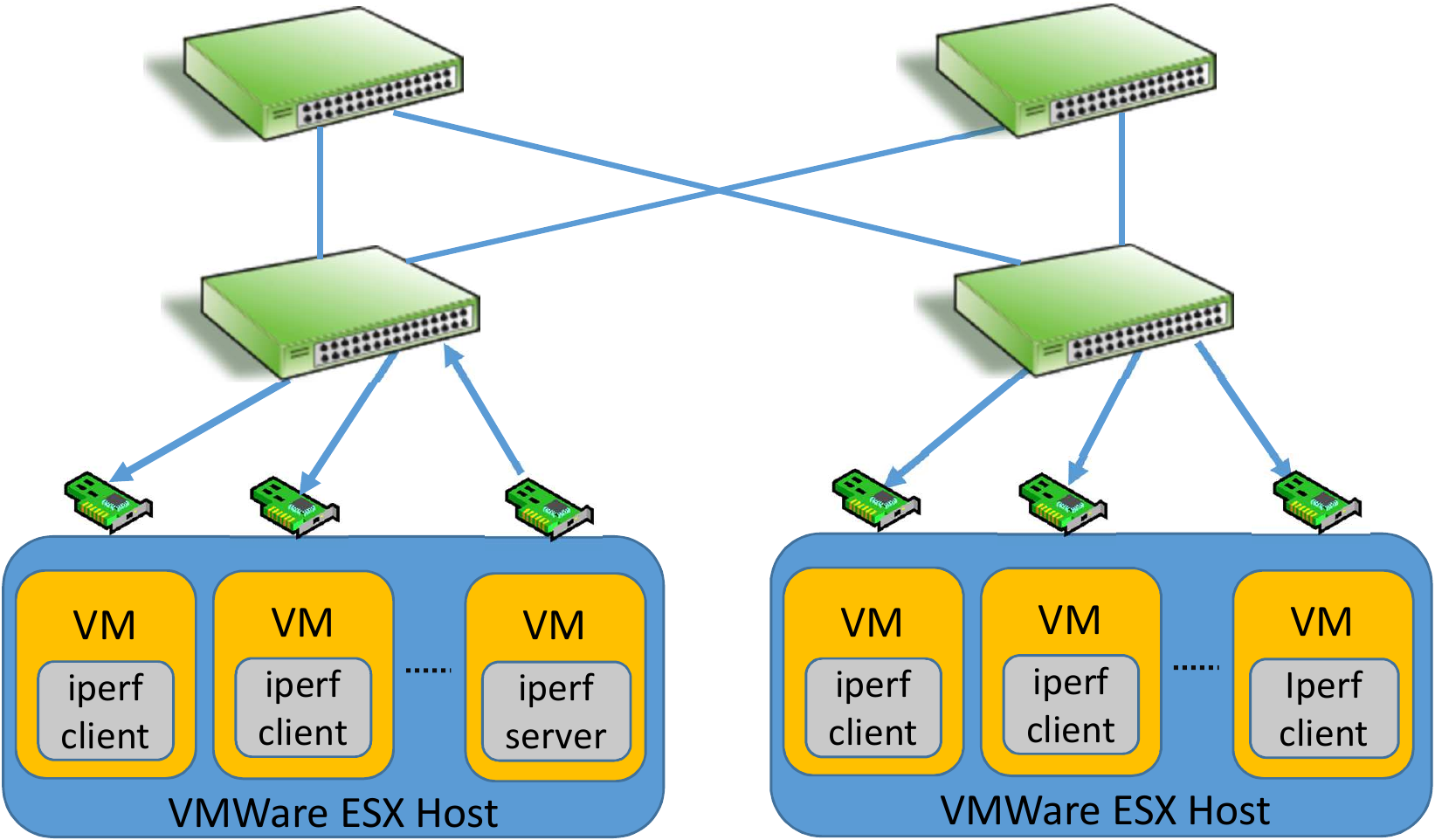}
\caption{Test Bed Topology}
\label{fig_testbed}
\end{figure}

\noindent\textbf{Choose the Curve Scaling Factor:} The fitting error is inversely related with the scaling factor $r$. We test the achievable throughput when 100 flows compete for the bottleneck link with different scaling factors and find that the throughput is 7.7Gbps when $r=1$. A larger $r$ leads to higher throughput. For example, when $r=4$, the throughput becomes 8.8Gbps. 
However, a larger $r$ also increases the QCD granularity, which violates our fine adjustment principle. Indeed, as shown in Figure~\ref{fig_scaling}, although the maximum queue depth does not increase by much when $n$ becomes larger, the queue depth oscillation becomes more intense.  As a tradeoff, we choose to use $r=4$ and $\texttt{P\_max} = 0.2$.

\begin{figure}[!tbh]
\centering
\includegraphics[width=.5\columnwidth]{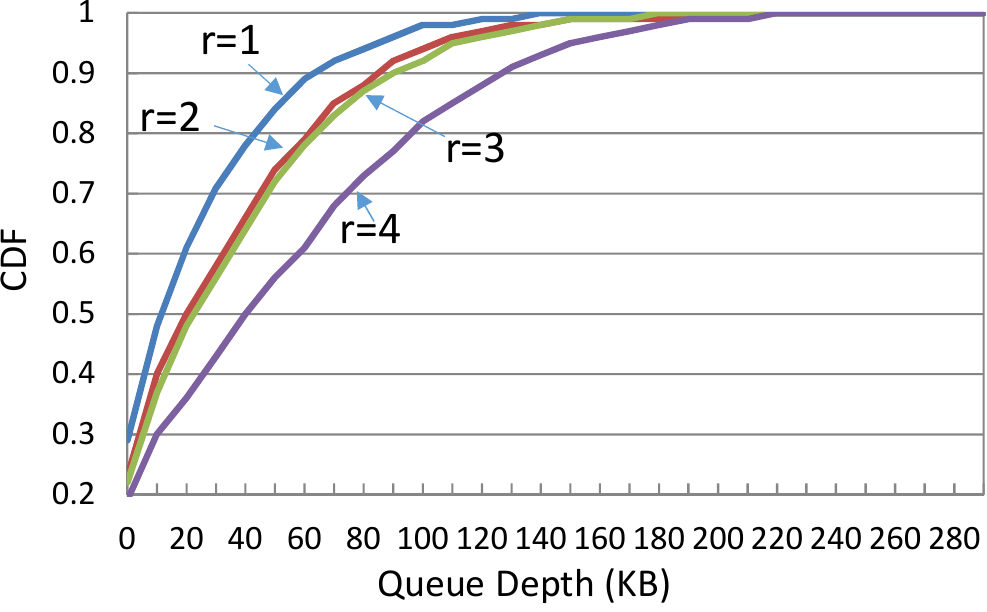}
\caption{Queue Depth on Scaling Factors}
\label{fig_scaling}
\end{figure}

\noindent\textbf{Bandwidth Utilization and Queue Depth:}  Figure~\ref{fig_dctcptest} shows that,  due to buffer underflow, the theoretical $k$ value for DCTCP cannot guarantee full bottleneck bandwidth utilization when the number of flows is large. Increasing $k$ helps improve the bandwidth utilization, but when the flow number is large, the improvement diminishes. This is because a larger $k$ consumes more buffer and will eventually cause buffer overflow (see Figure~\ref{fig_dcvtb1}(b)).   

\begin{figure}[!tbh] 
\centering
\includegraphics[width=.5\columnwidth]{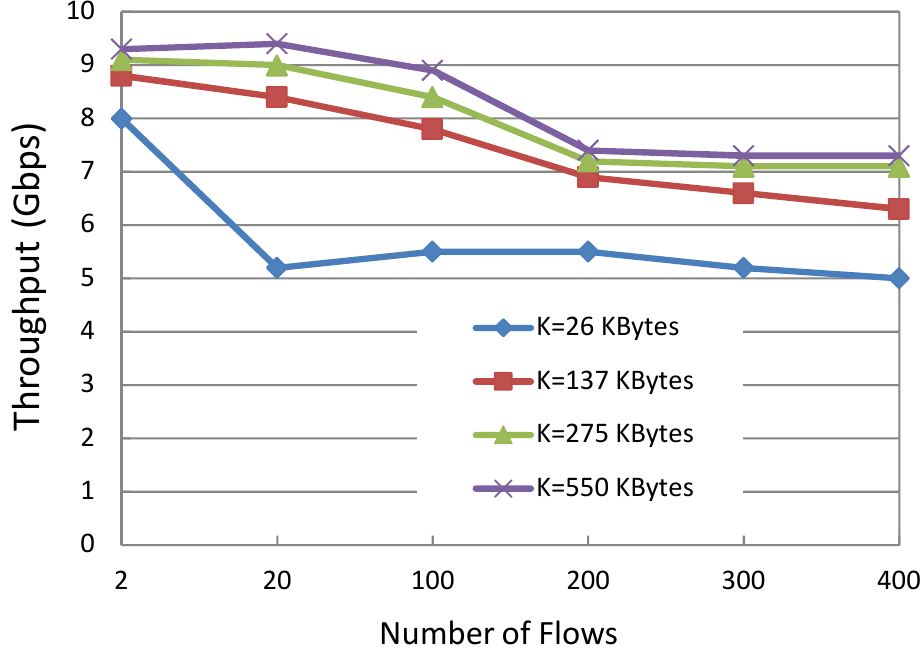}
\caption{DCTCP Throughput on $k$ and Flow Count}
\label{fig_dctcptest}
\end{figure}

When the number of flows is large, TBTCP can improve the bottleneck link throughput by up to 15\% compared with DCTCP. At same time, TBTCP reduces the maximum queue depth by up to 80\%, as shown in Figure~\ref{fig_dcvtb1}.

\begin{figure}[!tbh] 
\centering
\includegraphics[width=.75\columnwidth]{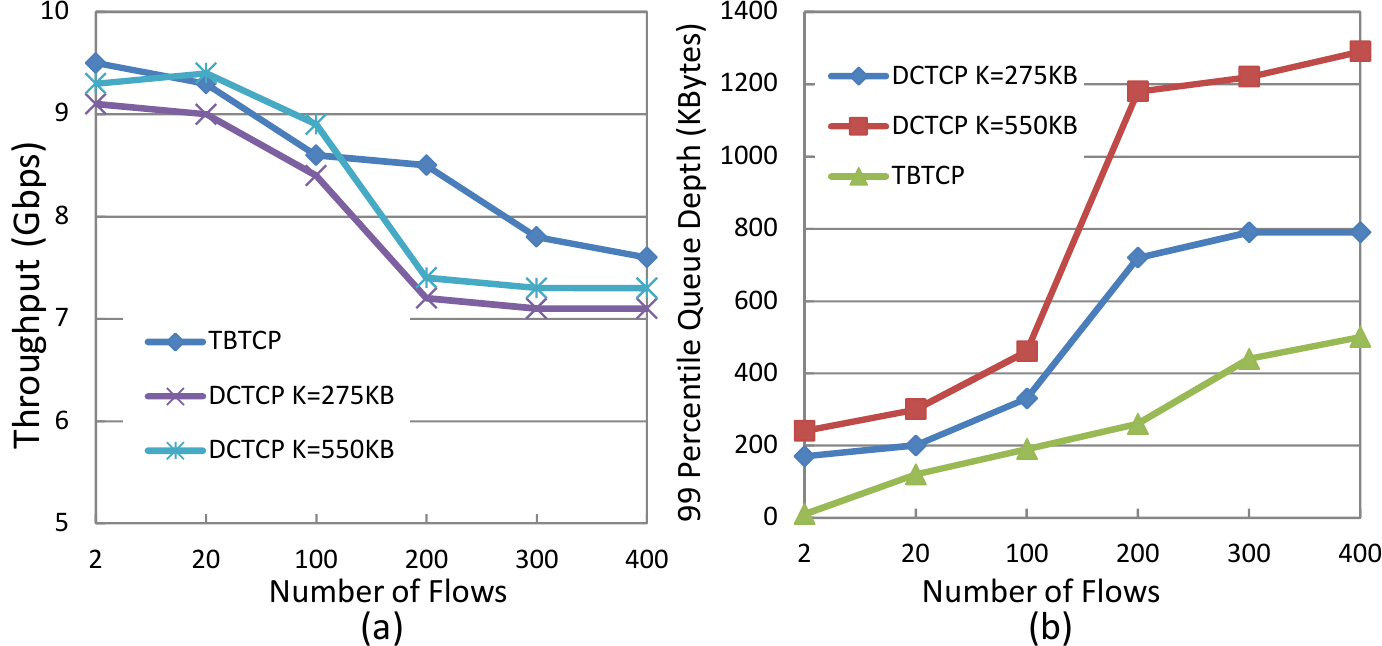}
\caption{TBTCP vs. DCTCP on BW and Queue}
\label{fig_dcvtb1}
\end{figure}

In the following tests, we set $k=$ 275KB for DCTCP as a tradeoff of throughput and queue depth.

\noindent\textbf{Flow Completion Time (FCT):} To emulate the real network traffic, we configure 80\% of the current flows to be shorter than 100KB and the remaining 20\% to be longer than 100KB. The flow size ranges from 4KB to 5MB. We run the test until we get 50 FCT values for each flow length. The average FCT values for TBTCP and DCTCP  are shown in Figure~\ref{fig_fct}(a).

\begin{figure}[!tbh] 
\centering
\includegraphics[width=.75\columnwidth]{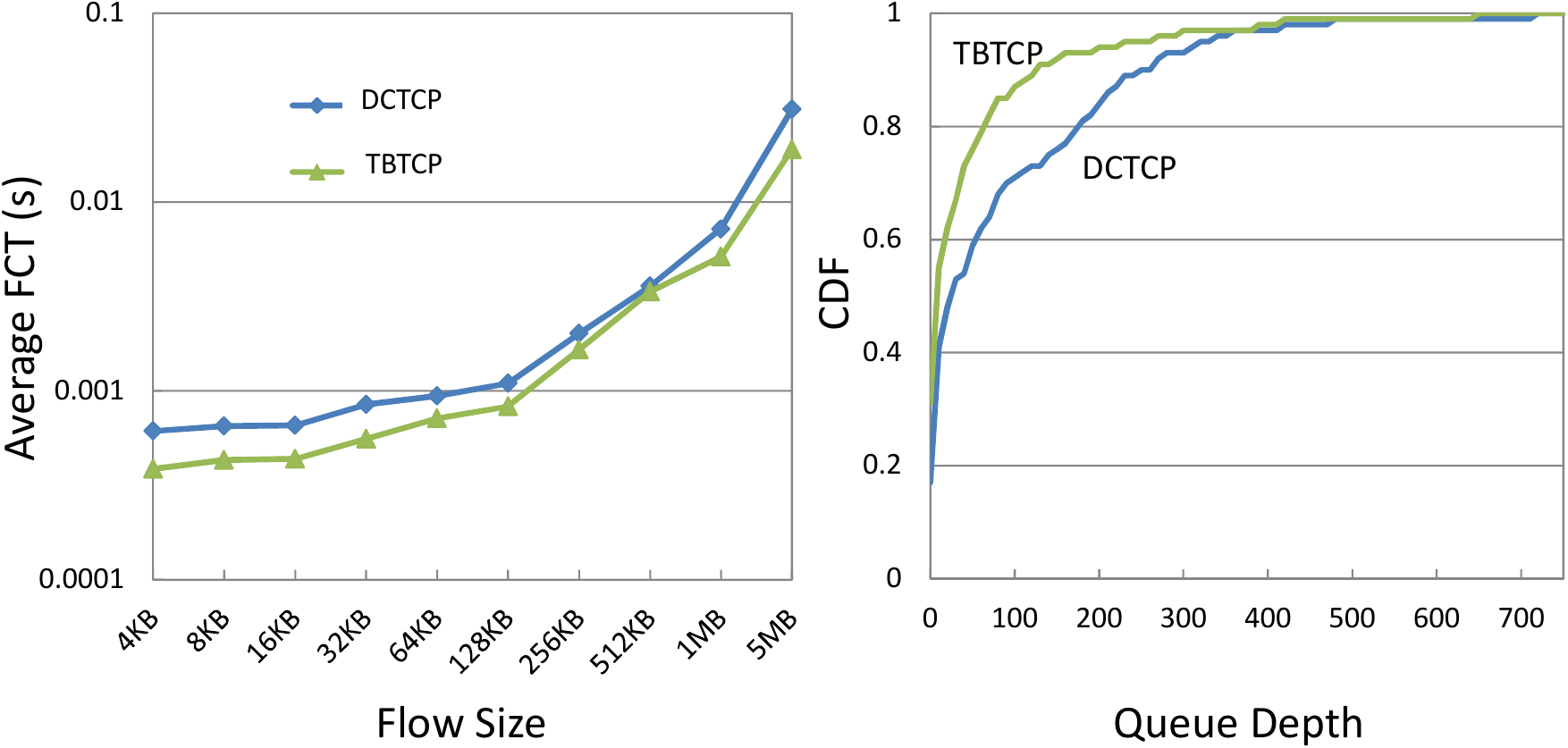}
\caption{TBTCP vs. DCTCP on FCT}
\label{fig_fct}
\end{figure}

The test shows that, compared with DCTCP,  TBTCP improves FCT for all flow sizes. TBTCP reduces the average FCT of short flows ($\leq$512KB) by about 200us (which accounts for 7\% to 37\% improvements) and the average FCT of long flows ($>$1MB) by 2-10ms (which account for 29\% to 39\% improvements). The significant FCT improvement is attributed to TBTCP's shallow and stable queue, as shown in Figure~\ref{fig_fct}(b). This result also confirms that the slower convergence time of TBTCP does not hurt the FCT performance.

\noindent\textbf{Fairness and Convergence Time:} We test bandwidth fairness and convergence time for both TBTCP and DCTCP. We start a new flow every 20 seconds until there are four flows. We then stop a flow every 20 seconds until only one flow is left. The test results are shown in Figure~\ref{fig_dcvtb2}.

\begin{figure}[!bt] 
\centering
\includegraphics[width=.75\columnwidth]{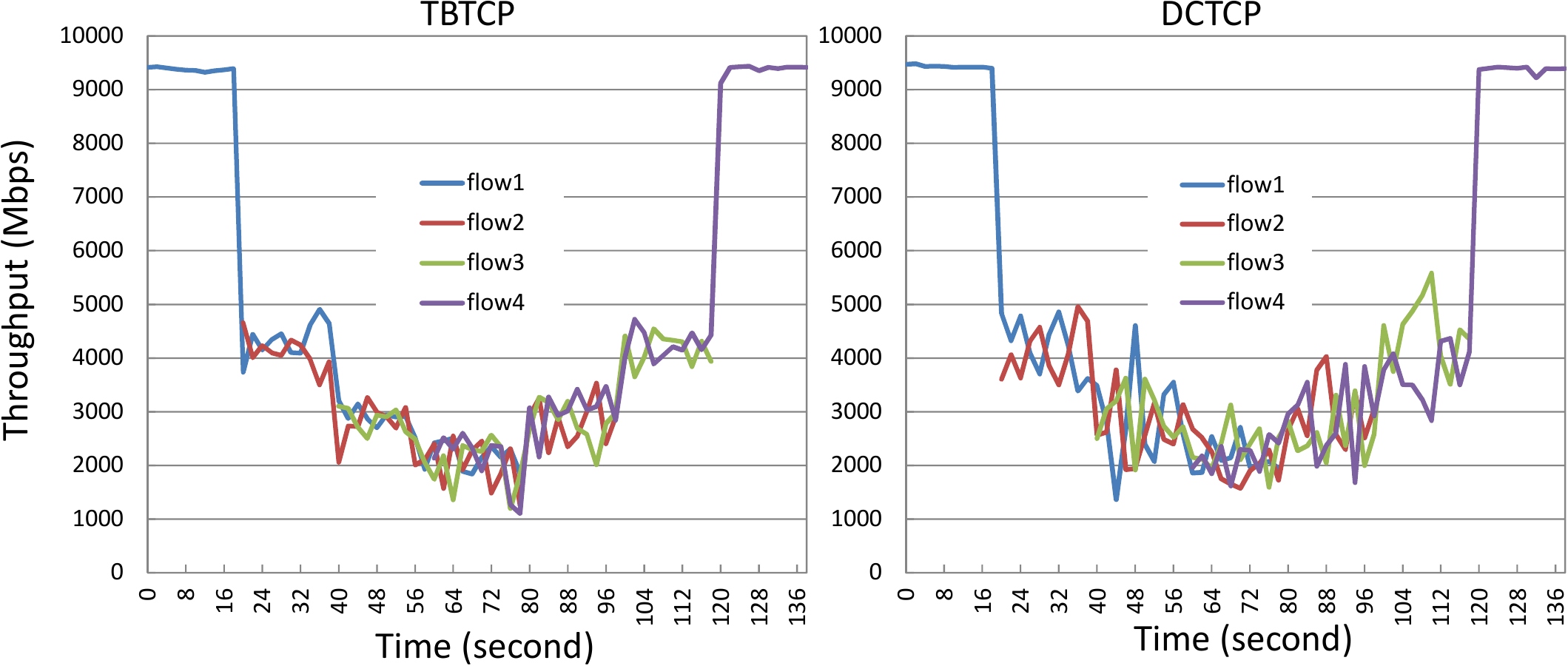}
\caption{TBTCP vs. DCTCP on Fairness}
\label{fig_dcvtb2}
\end{figure}

While both algorithms allow flows to share bandwidth fairly, TBTCP's flow throughput fluctuates less than DCTCP's.  We compare the bandwidth convergence time at each transition point (i.e., when a new flow starts or an old flow stops) and find that TBTCP and DCTCP are comparable when there is more than one flow. The two exceptions are: when the second flow starts, it takes TBTCP 0.3 seconds to converge to its share of bandwidth while it takes DCTCP 0.2 seconds; when the second-to-last flow stops, the last flow takes TBTCP 1.1 seconds to converge to full bandwidth while it takes DCTCP 0.4 seconds. 

Although DCTCP converges faster in general, the convergence time difference becomes negligible as the number of concurrent flows increases. 


%% file: summary.tex
\section{Summary and Discussion}~\label{summary}


TBTCP is a pure host TCP-based congestion control algorithm on a par with DCTCP. We did not compare TBTCP with other more recent data center congestion control algorithms because 
all those algorithms need network device modifications other than simple host OS updates.  
The design goal of TBTCP is to alleviate the buffer sizing challenge of high-density commodity switches while achieving better performance than DCTCP. 

A large buffer leads to poor FCT and strains the hardware; an unattended small buffer suffers both underflow and overflow. Therefore, the ideal strategy is to prevent congestion in the first place with fine control over a small buffer.
QCD duly eliminates the residual packets in buffer and RAI gently explores any available bandwidth. Jointly, QCD and RAI help TBTCP work close to the optimal point.
The TBTCP's bottleneck queue depth is not coupled with BDP but merely a small ratio of the flow count, which enables a big buffer size reduction.



We retrospect the rationality of TBTCP as follows: First, theoretical analysis tells us exactly how much we should adjust the window size to fit the network condition. The RED-based ECN is proactive through ``early'' notifications, which prevents the buildup of queues in advance. Meanwhile, the ``random'' notifications disperse the window adjustments to multiple flows. So, 1) the large flows are likely to reduce their windows more; 2) each notification only reduces the window of an affected flow by one; and 3) the aggregated adjustments exactly cancel the residual packets in buffer. 

Second, the conventional additive window increase is an overreaction in DCN where RTT is small. It causes bandwidth overshoot and queue buildup, which in turn leads to more congestion notifications. The violent queue oscillation demands larger buffer and fails to maintain a predictable performance.  RAI's slower window recovery largely avoids this problem and keeps the queue depth stable. The smaller recovery step is compensated by the facts that the window reduction size is smaller and large flows are more likely to receive a window reduction.

Third, in DCN, FCT is the most important performance indicator for applications. As the link bandwidth is maximized, RTT becomes the most influential factor affecting FCT. Due to the small hop count in DCN, RTT is mainly determined by the queuing delay. TBTCP's small buffer design therefore guarantees better FCT.

Fourth, 
The hierarchical topology of DCN makes a flow's RTT value drop into a few distinct bins. For example, in a 3-tier fat tree, flows may have different hop count of 1, 3, or 5. Conceivably, one flow's RTT can be multiple times of another's. The RED approach taken by TBTCP properly addresses the RTT fairness issue by introducing early and biased punishment to flows with smaller RTT when they try to grab more bandwidth. 

At last, QCD and RAI are both simple operations which demand less host computing power than DCTCP and other sophisticated congestion control algorithms. The required OS network protocol stack modification is minuscule.  

Nevertheless, TBTCP does not gain these benefits without any cost. We acknowledge that TBTCP retains some issues of DCTCP. For example, when TBTCP flows coexist with flows running other congestion control algorithms, the conservative use of buffer may negatively affect TBTCP flows.
However, TBTCP, just like DCTCP,  is intended to be deployed in a confined network environments with a single administration entity. It would be not difficult to deploy TBTCP exclusively throughout a data center,  
as the success of DCTCP has already proved the feasibility.

We also admit that TBTCP's flow bandwidth convergence speed is inferior to DCTCP's. However, the convergence time is an indirect performance indicator. 
TBTCP outperforms DCTCP is almost every other direct performance indicator, showing that the advantages of TBTCP dominate.

\noindent\textbf{Hardware Consideration:} While TBTCP works on commodity hardware, we suggest two possible implementation improvements. First, if we can configure the probability of each step in the chip's RED function, we can fit the ideal marking probability curve without resorting to the curve scaling approach. This way, we can stick to the original QCD design: a flow only reduces its window by one for each ECN. Since the chip has already supported the step function, we speculate the capability is ready and the vendor just needs to open the API to switch users. We will confirm this with the chip vendor and suggest this enhancement if the chip is incapable of such configuration.  

Second, as in the case of DCTCP, we find the burstiness introduced by the LSO feature on NIC to be harmful to TBTCP's performance. A better burst control on NIC with packet pacing will be very helpful. This can be done through a customized NIC with extra packet processing capability. The smart NIC with an FPGA can be a perfect platform~\cite{smartnic}. 

Our future work will pursue these directions to see how close the real implementation can approach the ideal performance of the TBTCP algorithm.       

%% file: related.tex
\section{Related Work}~\label{related}

Many new DCN architectures and algorithms have been published in recent years. Unlike TBTCP, some schemes are network-based, which require new switch hardware to cope with congestion and make better use of network resource (e.g., NDP~\cite{NDP} assumes a specialized switch for packet trimming and queuing). These schemes often make little or no assumption on end-host's transport protocols. For example, CONGA~\cite{conga} actively measures the congestion level of all paths for each flow and dispatches flowlets to less congested paths for load balancing; LetFlow~\cite{letflow} shows that random flowlet-based load balancing performs surprisingly well in reality; RackCC~\cite{rackcc} conducts rack-level congestion control directly on switches. A well-behaved TCP, such as TBTCP, will work with these algorithms to provide even better performance.  

Data center operators have a strong incentive to keep complexity at edge and make use of commodity switches as the network fabric~\cite{fabricsdn}, because low cost and instant deployability are essential for data center services. The most straightforward approach is to deploy a better TCP congestion control algorithm. Unfortunately, TCP variants such as New Reno~\cite{newreno} and CUBIC~\cite{cubic} cannot meet the stringent performance requirements in DCN, partially due to their design focus on the Wide Area Network (WAN) environments. DCTCP~\cite{dctcp} uses ECN rather than packet drop as a congestion signal, and applies gentler window adjustment. The congestion sensing and reaction is much faster, and catastrophic packet drops are mostly avoided. As a result, the FCT performance is improved. Further optimizations (e.g., DPP~\cite{dpp} and PIAS~\cite{pias}) distinguish long and short flows by giving short flows higher forwarding priority at switches so as to imporve flow FCT performance. Similarly, these optimizations can be applied on TBTCP to boost its performance.   

Instead of adjusting the sending window based on congestion signals, some end-host congestion control algorithms determine the flow sending rates based on the RTT measurements. For example, BBR~\cite{bbr} uses RTT measurements to estimate the network BDP and pace packets at the highest possible sending rate. TIMELY~\cite{timely} applies a similar idea but bypasses OS-kernel and offloads the RTT measurement and traffic pacing to NIC. QUIC~\cite{quic} implements new flow control algorithms in UDP to surpass TCP's performance for particular applications such as HTTP. Other new transport algorithms have been proposed to replace TCP completely (e.g, PCC~\cite{pcc}, pHost~\cite{phost}, ExpressPass~\cite{expresspass}, and HOMA~\cite{homa}). 

More flow related information is useful for better transport performance. In order to gain better performance than TCP, some algorithms rely on application layer information, such as flow deadline or priority, for flow scheduling (e.g., D$^2$TCP~\cite{d2tcp}, D$^3$~\cite{d3}, PDQ~\cite{pdq}, and pFabric~\cite{pfabric}). In addition to the availability assumption of the application layer information, some of these algorithms also need switch arbitration (e.g., D$^3$~\cite{d3} and PDQ~\cite{pdq}) for flow rate allocation. 

Due to the dominant scale of TCP deployment and applications relying on TCP, it may be challenging to deploy such algorithms in reality. However, there is a noticeable exception. For High Performance Computing (HPC), Artificial Intelligence (AI), and storage disaggregation applications in data centers, Remote Direct Memory Access (RDMA), due to its low latency and low CPU overhead, is gaining its momentum by providing the full stack networking solution~\cite{dcqcn}. For a low system cost, it is desirable to avoid proprietary hardware and apply commodity Ethernet technologies to the DCN fabric. RDMA over Converged Ethernet (RoCE) has been an active research area. Similar to DCTCP, DCQCN uses ECN as signal for flow rate control~\cite{dcqcn}. Since RDMA requires a lossless network fabric, the link based flow control, PFC~\cite{pfc}, is used to prevent packet drop due to buffer overflow in networks. The derived issues such as deadlock and head-of-line blocking are addressed in~\cite{rocescale, dcdeadlock}. TBTCP can be used as an orthogonal optimization to mitigate these issues. PFC is less likely to be triggered when using TBTCP, since TBTCP applies more graceful window adjustments and maintains low and stable buffer occupancy. 
     
Some other algorithms combine in-network scheduling/load balancing and end-host rate adjustment together to achieve better performance. PASE~\cite{pase} proposes a distributed framework while Fastpass~\cite{fastpass} takes a centralized approach. In either case, TBTCP can be used for end-host rate control, which is complementary to in-network optimizations.

%% file: main.bbl

\begin{thebibliography}{43}


\ifx \showCODEN    \undefined \def \showCODEN     #1{\unskip}     \fi
\ifx \showDOI      \undefined \def \showDOI       #1{#1}\fi
\ifx \showISBNx    \undefined \def \showISBNx     #1{\unskip}     \fi
\ifx \showISBNxiii \undefined \def \showISBNxiii  #1{\unskip}     \fi
\ifx \showISSN     \undefined \def \showISSN      #1{\unskip}     \fi
\ifx \showLCCN     \undefined \def \showLCCN      #1{\unskip}     \fi
\ifx \shownote     \undefined \def \shownote      #1{#1}          \fi
\ifx \showarticletitle \undefined \def \showarticletitle #1{#1}   \fi
\ifx \showURL      \undefined \def \showURL       {\relax}        \fi
\providecommand\bibfield[2]{#2}
\providecommand\bibinfo[2]{#2}
\providecommand\natexlab[1]{#1}
\providecommand\showeprint[2][]{arXiv:#2}

\bibitem[\protect\citeauthoryear{Alizadeh, Edsall, Dharmapurikar, Vaidyanathan,
  Chu, Fingerhut, Lam, Matus, Pan, Yadav, and Varghese}{Alizadeh
  et~al\mbox{.}}{2014}]%
        {conga}
\bibfield{author}{\bibinfo{person}{Mohammad Alizadeh}, \bibinfo{person}{Tom
  Edsall}, \bibinfo{person}{Sarang Dharmapurikar}, \bibinfo{person}{Ramanan
  Vaidyanathan}, \bibinfo{person}{Kevin Chu}, \bibinfo{person}{Andy Fingerhut},
  \bibinfo{person}{Vinh~The Lam}, \bibinfo{person}{Francis Matus},
  \bibinfo{person}{Rong Pan}, \bibinfo{person}{Navindra Yadav}, {and}
  \bibinfo{person}{George Varghese}.} \bibinfo{year}{2014}\natexlab{}.
\newblock \showarticletitle{{CONGA: Distributed Congestion-Aware Load Balancing
  for Datacenters}}. In \bibinfo{booktitle}{{\em ACM SIGCOMM}}.
\newblock


\bibitem[\protect\citeauthoryear{Alizadeh, Greenberg, Maltz, Padhye, Patel,
  Prabhakar, Sengupta, and Sridharan}{Alizadeh et~al\mbox{.}}{2010}]%
        {dctcp}
\bibfield{author}{\bibinfo{person}{Mohammad Alizadeh}, \bibinfo{person}{Albert
  Greenberg}, \bibinfo{person}{David~A. Maltz}, \bibinfo{person}{Jitendra
  Padhye}, \bibinfo{person}{Parveen Patel}, \bibinfo{person}{Balaji Prabhakar},
  \bibinfo{person}{Sudipta Sengupta}, {and} \bibinfo{person}{Murari
  Sridharan}.} \bibinfo{year}{2010}\natexlab{}.
\newblock \showarticletitle{{Data Center TCP (DCTCP)}}. In
  \bibinfo{booktitle}{{\em ACM SIGCOMM}}.
\newblock


\bibitem[\protect\citeauthoryear{Alizadeh, Javanmard, and Prabhakar}{Alizadeh
  et~al\mbox{.}}{2011}]%
        {dctcpana}
\bibfield{author}{\bibinfo{person}{Mohammad Alizadeh}, \bibinfo{person}{Adel
  Javanmard}, {and} \bibinfo{person}{Balaji Prabhakar}.}
  \bibinfo{year}{2011}\natexlab{}.
\newblock \showarticletitle{Analysis of DCTCP: Stability, Convergence, and
  Fairness}. In \bibinfo{booktitle}{{\em ACM SIGMETRICS}}.
\newblock


\bibitem[\protect\citeauthoryear{Alizadeh, Yang, Sharif, Katti, McKeown,
  Prabhakar, and Shenker}{Alizadeh et~al\mbox{.}}{2013}]%
        {pfabric}
\bibfield{author}{\bibinfo{person}{Mohammad Alizadeh}, \bibinfo{person}{Shuang
  Yang}, \bibinfo{person}{Milad Sharif}, \bibinfo{person}{Sachin Katti},
  \bibinfo{person}{Nick McKeown}, \bibinfo{person}{Balaji Prabhakar}, {and}
  \bibinfo{person}{Scott Shenker}.} \bibinfo{year}{2013}\natexlab{}.
\newblock \showarticletitle{{pFabric: Minimal Near-Optimal Datacenter
  Transport}}. In \bibinfo{booktitle}{{\em ACM SIGCOMM}}.
\newblock


\bibitem[\protect\citeauthoryear{Altman, Barakat, and Laborde}{Altman
  et~al\mbox{.}}{2000}]%
        {rttfair}
\bibfield{author}{\bibinfo{person}{Eitan Altman}, \bibinfo{person}{Chadi
  Barakat}, {and} \bibinfo{person}{Emmanuel Laborde}.}
  \bibinfo{year}{2000}\natexlab{}.
\newblock \showarticletitle{{Fairness analysis of TCP/IP}}. In
  \bibinfo{booktitle}{{\em 39th IEEE Conference on Decision and Control}}.
\newblock


\bibitem[\protect\citeauthoryear{Bai, Chen, Chen, Han, Tian, and Wang}{Bai
  et~al\mbox{.}}{2015}]%
        {pias}
\bibfield{author}{\bibinfo{person}{Wei Bai}, \bibinfo{person}{Li Chen},
  \bibinfo{person}{Kai Chen}, \bibinfo{person}{Dongsu Han},
  \bibinfo{person}{Chen Tian}, {and} \bibinfo{person}{Hao Wang}.}
  \bibinfo{year}{2015}\natexlab{}.
\newblock \showarticletitle{{Information-Agnostic Flow Scheduling for Commodity
  Data Centers}}. In \bibinfo{booktitle}{{\em USENIX NSDI}}.
\newblock


\bibitem[\protect\citeauthoryear{Bai, Chen, Chen, and Wu}{Bai
  et~al\mbox{.}}{2016}]%
        {mqecn}
\bibfield{author}{\bibinfo{person}{Wei Bai}, \bibinfo{person}{Li Chen},
  \bibinfo{person}{Kai Chen}, {and} \bibinfo{person}{Haitao Wu}.}
  \bibinfo{year}{2016}\natexlab{}.
\newblock \showarticletitle{{Enabling ECN in Multi-Service Multi-Queue Data
  Centers}}. In \bibinfo{booktitle}{{\em USENIX NSDI}}.
\newblock


\bibitem[\protect\citeauthoryear{{Broadcom}}{{Broadcom}}{2012}]%
        {bcmguide}
\bibfield{author}{\bibinfo{person}{{Broadcom}}.}
  \bibinfo{year}{2012}\natexlab{}.
\newblock \bibinfo{title}{{BCM56840 Programmer's Register Reference Guide}}.
\newblock   (\bibinfo{year}{2012}).
\newblock


\bibitem[\protect\citeauthoryear{Cardwell, Cheng, Gunn, Yeganeh, and
  Jacobson}{Cardwell et~al\mbox{.}}{2016}]%
        {bbr}
\bibfield{author}{\bibinfo{person}{Neal Cardwell}, \bibinfo{person}{Yuchung
  Cheng}, \bibinfo{person}{C.~Stephen Gunn}, \bibinfo{person}{Soheil~Hassas
  Yeganeh}, {and} \bibinfo{person}{Van Jacobson}.}
  \bibinfo{year}{2016}\natexlab{}.
\newblock \showarticletitle{{BBR: Congestion-Based Congestion Control}}.
\newblock \bibinfo{journal}{{\em acmqueue\/}}  \bibinfo{volume}{14}
  (\bibinfo{date}{December} \bibinfo{year}{2016}).
\newblock
Issue 5.


\bibitem[\protect\citeauthoryear{Casado, Koponen, Shenker, and
  Tootoonchian}{Casado et~al\mbox{.}}{2012}]%
        {fabricsdn}
\bibfield{author}{\bibinfo{person}{Martín Casado}, \bibinfo{person}{Teemu
  Koponen}, \bibinfo{person}{Scott Shenker}, {and} \bibinfo{person}{Amin
  Tootoonchian}.} \bibinfo{year}{2012}\natexlab{}.
\newblock \showarticletitle{{Fabric: A Retrospective on Evolving SDN}}. In
  \bibinfo{booktitle}{{\em ACM SIGCOMM HotSDN Workshop}}.
\newblock


\bibitem[\protect\citeauthoryear{Chen, Griffith, Liu, Katz, and Joseph}{Chen
  et~al\mbox{.}}{2009}]%
        {incast}
\bibfield{author}{\bibinfo{person}{Yanpei Chen}, \bibinfo{person}{Rean
  Griffith}, \bibinfo{person}{Junda Liu}, \bibinfo{person}{Randy~H. Katz},
  {and} \bibinfo{person}{Anthony~D. Joseph}.} \bibinfo{year}{2009}\natexlab{}.
\newblock \showarticletitle{Understanding TCP Incast Throughput Collapse in
  Datacenter Networks}. In \bibinfo{booktitle}{{\em 1st ACM Workshop on
  Research on Enterprise Networking (WREN)}}.
\newblock


\bibitem[\protect\citeauthoryear{Cho, Han, and Jang}{Cho et~al\mbox{.}}{2016}]%
        {expresspass}
\bibfield{author}{\bibinfo{person}{Inho Cho}, \bibinfo{person}{Dongsu Han},
  {and} \bibinfo{person}{Keon Jang}.} \bibinfo{year}{2016}\natexlab{}.
\newblock \showarticletitle{{ExpressPass: End-to-End Credit-based Congestion
  Control for Datacenters}}.
\newblock \bibinfo{journal}{{\em CoRR\/}}  \bibinfo{volume}{abs/1610.04688}
  (\bibinfo{year}{2016}).
\newblock


\bibitem[\protect\citeauthoryear{{Cisco}}{{Cisco}}{2015}]%
        {dpp}
\bibfield{author}{\bibinfo{person}{{Cisco}}.} \bibinfo{year}{2015}\natexlab{}.
\newblock \bibinfo{title}{{Dynamic Packet Prioritization, Nexus 9000 Series
  Switches}}.
\newblock   (\bibinfo{year}{2015}).
\newblock


\bibitem[\protect\citeauthoryear{Cronkite-Ratcliff, Bergman, Vargaftik, Ravi,
  McKeown, Abraham, and Keslassy}{Cronkite-Ratcliff et~al\mbox{.}}{2016}]%
        {vcc}
\bibfield{author}{\bibinfo{person}{Bryce Cronkite-Ratcliff},
  \bibinfo{person}{Aran Bergman}, \bibinfo{person}{Shay Vargaftik},
  \bibinfo{person}{Madhusudhan Ravi}, \bibinfo{person}{Nick McKeown},
  \bibinfo{person}{Ittai Abraham}, {and} \bibinfo{person}{Isaac Keslassy}.}
  \bibinfo{year}{2016}\natexlab{}.
\newblock \showarticletitle{{Virtualized Congestion Control}}. In
  \bibinfo{booktitle}{{\em ACM SIGCOMM}}.
\newblock


\bibitem[\protect\citeauthoryear{Dong, Li, Zarchy, Godfrey, and Schapira}{Dong
  et~al\mbox{.}}{2015}]%
        {pcc}
\bibfield{author}{\bibinfo{person}{Mo Dong}, \bibinfo{person}{Qingxi Li},
  \bibinfo{person}{Doron Zarchy}, \bibinfo{person}{P.~Brighten Godfrey}, {and}
  \bibinfo{person}{Michael Schapira}.} \bibinfo{year}{2015}\natexlab{}.
\newblock \showarticletitle{{PCC: Re-architecting Congestion Control for
  Consistent High Performance}}. In \bibinfo{booktitle}{{\em USENIX NSDI}}.
\newblock


\bibitem[\protect\citeauthoryear{Dukkipati and McKeown}{Dukkipati and
  McKeown}{2006}]%
        {fct}
\bibfield{author}{\bibinfo{person}{Nandita Dukkipati} {and}
  \bibinfo{person}{Nick McKeown}.} \bibinfo{year}{2006}\natexlab{}.
\newblock \showarticletitle{{Why Flow-Completion Time is the Right Metric for
  Congestion Control}}.
\newblock \bibinfo{journal}{{\em ACM SIGCOMM CCR\/}} \bibinfo{volume}{36},
  \bibinfo{number}{1} (\bibinfo{date}{January} \bibinfo{year}{2006}).
\newblock


\bibitem[\protect\citeauthoryear{Gao, Narayan, Kumar, Agarwal, Ratnasamy, and
  Shenker}{Gao et~al\mbox{.}}{2015}]%
        {phost}
\bibfield{author}{\bibinfo{person}{Peter~X. Gao}, \bibinfo{person}{Akshay
  Narayan}, \bibinfo{person}{Gautam Kumar}, \bibinfo{person}{Rachit Agarwal},
  \bibinfo{person}{Sylvia Ratnasamy}, {and} \bibinfo{person}{Scott Shenker}.}
  \bibinfo{year}{2015}\natexlab{}.
\newblock \showarticletitle{{pHost: Distributed Near-Optimal Datacenter
  Transport Over Commodity Network Fabric}}. In \bibinfo{booktitle}{{\em ACM
  CoNEXT}}.
\newblock


\bibitem[\protect\citeauthoryear{Greenberg}{Greenberg}{2015}]%
        {smartnic}
\bibfield{author}{\bibinfo{person}{Albert Greenberg}.}
  \bibinfo{year}{2015}\natexlab{}.
\newblock \showarticletitle{{SDN for the Cloud}}. In \bibinfo{booktitle}{{\em
  ACM SIGCOMM}}.
\newblock


\bibitem[\protect\citeauthoryear{Guo, Wu, Deng, Soni, Ye, Padhye, and
  Lipshteyn}{Guo et~al\mbox{.}}{2016}]%
        {rocescale}
\bibfield{author}{\bibinfo{person}{Chuanxiong Guo}, \bibinfo{person}{Haitao
  Wu}, \bibinfo{person}{Zhong Deng}, \bibinfo{person}{Gaurav Soni},
  \bibinfo{person}{Jianxi Ye}, \bibinfo{person}{Jitendra Padhye}, {and}
  \bibinfo{person}{Marina Lipshteyn}.} \bibinfo{year}{2016}\natexlab{}.
\newblock \showarticletitle{{RDMA over Commodity Ethernet at Scale}}. In
  \bibinfo{booktitle}{{\em ACM SIGCOMM}}.
\newblock


\bibitem[\protect\citeauthoryear{Handley, Raiciu, Agache, Voinescu, Moore,
  Antichi, and W\'{o}jcik}{Handley et~al\mbox{.}}{2017}]%
        {NDP}
\bibfield{author}{\bibinfo{person}{Mark Handley}, \bibinfo{person}{Costin
  Raiciu}, \bibinfo{person}{Alexandru Agache}, \bibinfo{person}{Andrei
  Voinescu}, \bibinfo{person}{Andrew~W. Moore}, \bibinfo{person}{Gianni
  Antichi}, {and} \bibinfo{person}{Marcin W\'{o}jcik}.}
  \bibinfo{year}{2017}\natexlab{}.
\newblock \showarticletitle{Re-architecting Datacenter Networks and Stacks for
  Low Latency and High Performance}. In \bibinfo{booktitle}{{\em ACM SIGCOMM}}.
\newblock


\bibitem[\protect\citeauthoryear{He, Rozner, Agarwal, Gu, Felter, Carter, and
  Akella}{He et~al\mbox{.}}{2016}]%
        {acdctcp}
\bibfield{author}{\bibinfo{person}{Keqiang He}, \bibinfo{person}{Eric Rozner},
  \bibinfo{person}{Kanak Agarwal}, \bibinfo{person}{Yu Gu},
  \bibinfo{person}{Wes Felter}, \bibinfo{person}{John Carter}, {and}
  \bibinfo{person}{Aditya Akella}.} \bibinfo{year}{2016}\natexlab{}.
\newblock \showarticletitle{{AC/DC TCP: Virtual Congestion Control Enforcement
  for Datacenter Networks}}. In \bibinfo{booktitle}{{\em ACM SIGCOMM}}.
\newblock


\bibitem[\protect\citeauthoryear{Henderson, Floyd, Gurtov, and
  Nishida}{Henderson et~al\mbox{.}}{2012}]%
        {newreno}
\bibfield{author}{\bibinfo{person}{Tom Henderson}, \bibinfo{person}{Sally
  Floyd}, \bibinfo{person}{Andrei Gurtov}, {and} \bibinfo{person}{Yoshifumi
  Nishida}.} \bibinfo{year}{2012}\natexlab{}.
\newblock \bibinfo{title}{{IETF RFC 6582: The NewReno Modification to TCP's
  Fast Recovery Algorithm}}.
\newblock   (\bibinfo{year}{2012}).
\newblock


\bibitem[\protect\citeauthoryear{Hong, Caesar, and Godfrey}{Hong
  et~al\mbox{.}}{2012}]%
        {pdq}
\bibfield{author}{\bibinfo{person}{Chi-Yao Hong}, \bibinfo{person}{Matthew
  Caesar}, {and} \bibinfo{person}{P.~Brighten Godfrey}.}
  \bibinfo{year}{2012}\natexlab{}.
\newblock \showarticletitle{{Finishing Flows Quickly with Preemptive
  Scheduling}}. In \bibinfo{booktitle}{{\em ACM SIGCOMM}}.
\newblock


\bibitem[\protect\citeauthoryear{Hu, Zhu, Cheng, Guo, Tan, Padhye, and Chen}{Hu
  et~al\mbox{.}}{2016}]%
        {dcdeadlock}
\bibfield{author}{\bibinfo{person}{Shuihai Hu}, \bibinfo{person}{Yibo Zhu},
  \bibinfo{person}{Peng Cheng}, \bibinfo{person}{Chuanxiong Guo},
  \bibinfo{person}{Kun Tan}, \bibinfo{person}{Jitendra Padhye}, {and}
  \bibinfo{person}{Kai Chen}.} \bibinfo{year}{2016}\natexlab{}.
\newblock \showarticletitle{{Deadlocks in Datacenter Networks: Why Do They
  Form, and How to Avoid Them}}. In \bibinfo{booktitle}{{\em ACM HotNets-XV}}.
\newblock


\bibitem[\protect\citeauthoryear{{IEEE}}{{IEEE}}{2011}]%
        {pfc}
\bibfield{author}{\bibinfo{person}{{IEEE}}.} \bibinfo{year}{2011}\natexlab{}.
\newblock \bibinfo{title}{{802.1Qbb: Priority-based Flow Control}}.
\newblock   (\bibinfo{year}{2011}).
\newblock


\bibitem[\protect\citeauthoryear{{Jim Roskind}}{{Jim Roskind}}{2013}]%
        {quic}
\bibfield{author}{\bibinfo{person}{{Jim Roskind}}.}
  \bibinfo{year}{2013}\natexlab{}.
\newblock \bibinfo{title}{{QUIC: Multiplexed Stream Transport over UDP}}.
\newblock   (\bibinfo{year}{2013}).
\newblock


\bibitem[\protect\citeauthoryear{{Jon Dugan et. al.}}{{Jon Dugan et.
  al.}}{2017}]%
        {iperf}
\bibfield{author}{\bibinfo{person}{{Jon Dugan et. al.}}}
  \bibinfo{year}{2017}\natexlab{}.
\newblock \bibinfo{title}{{iPerf: The ultimate speed test tool for TCP, UDP and
  SCTP}}.
\newblock   (\bibinfo{year}{2017}).
\newblock
\showURL{%
\url{https://iperf.fr/}}


\bibitem[\protect\citeauthoryear{Lee, Park, Jang, Moon, and Han}{Lee
  et~al\mbox{.}}{2015}]%
        {dx}
\bibfield{author}{\bibinfo{person}{Changhyun Lee}, \bibinfo{person}{Chunjong
  Park}, \bibinfo{person}{Keon Jang}, \bibinfo{person}{Sue Moon}, {and}
  \bibinfo{person}{Dongsu Han}.} \bibinfo{year}{2015}\natexlab{}.
\newblock \showarticletitle{Accurate Latency-based Congestion Feedback for
  Datacenters}. In \bibinfo{booktitle}{{\em USENIX Annual Technical Conference
  (USENIX ATC 15)}}.
\newblock


\bibitem[\protect\citeauthoryear{Mittal, Lam, Dukkipati, Blem, Wassel, Ghobadi,
  Vahdat, Wang, Wetherall, and Zats}{Mittal et~al\mbox{.}}{2015}]%
        {timely}
\bibfield{author}{\bibinfo{person}{Radhika Mittal}, \bibinfo{person}{Wihn~The
  Lam}, \bibinfo{person}{Nandita Dukkipati}, \bibinfo{person}{Emily Blem},
  \bibinfo{person}{Hassan Wassel}, \bibinfo{person}{Monia Ghobadi},
  \bibinfo{person}{Amin Vahdat}, \bibinfo{person}{Yaogong Wang},
  \bibinfo{person}{David Wetherall}, {and} \bibinfo{person}{Davis Zats}.}
  \bibinfo{year}{2015}\natexlab{}.
\newblock \showarticletitle{{TIMELY: RTT-based Congestion Control for the
  Datacenter}}. In \bibinfo{booktitle}{{\em ACM SIGCOMM}}.
\newblock


\bibitem[\protect\citeauthoryear{Montazeri, Ousterhout, Li, and
  Alizadeh}{Montazeri et~al\mbox{.}}{2018}]%
        {homa}
\bibfield{author}{\bibinfo{person}{Behnam Montazeri}, \bibinfo{person}{John
  Ousterhout}, \bibinfo{person}{Yilong Li}, {and} \bibinfo{person}{Mohammad
  Alizadeh}.} \bibinfo{year}{2018}\natexlab{}.
\newblock \showarticletitle{{Homa: A Receiver-Driven Low-Latency Transport
  Protocol Using Network Priorities}}. In \bibinfo{booktitle}{{\em review}}.
\newblock


\bibitem[\protect\citeauthoryear{Munir, Baig, Irteza, Qazi, Liu, and
  Dogar}{Munir et~al\mbox{.}}{2014}]%
        {pase}
\bibfield{author}{\bibinfo{person}{Ali Munir}, \bibinfo{person}{Ghufran Baig},
  \bibinfo{person}{Syed~M. Irteza}, \bibinfo{person}{Ihsan~A. Qazi},
  \bibinfo{person}{Alex~X. Liu}, {and} \bibinfo{person}{Fahad~R. Dogar}.}
  \bibinfo{year}{2014}\natexlab{}.
\newblock \showarticletitle{{Friends, not Foes – Synthesizing Existing
  Transport Strategies for Data Center Networks}}. In \bibinfo{booktitle}{{\em
  ACM SIGCOMM}}.
\newblock


\bibitem[\protect\citeauthoryear{Nichols and Jacobson}{Nichols and
  Jacobson}{2012}]%
        {codel}
\bibfield{author}{\bibinfo{person}{Kathleen Nichols} {and} \bibinfo{person}{Van
  Jacobson}.} \bibinfo{year}{2012}\natexlab{}.
\newblock \showarticletitle{{Controlling Queue Delay}}.
\newblock \bibinfo{journal}{{\em acmqueue\/}}  \bibinfo{volume}{10}
  (\bibinfo{date}{December} \bibinfo{year}{2012}).
\newblock
Issue 5.


\bibitem[\protect\citeauthoryear{Pan, Natarajan, Piglione, Prabhu, Subramanian,
  Baker, and VerSteeg}{Pan et~al\mbox{.}}{2013}]%
        {pie}
\bibfield{author}{\bibinfo{person}{Rong Pan}, \bibinfo{person}{Preethi
  Natarajan}, \bibinfo{person}{Chiara Piglione},
  \bibinfo{person}{Mythili~Suryanarayana Prabhu}, \bibinfo{person}{Vijay
  Subramanian}, \bibinfo{person}{Fred Baker}, {and} \bibinfo{person}{Bill
  VerSteeg}.} \bibinfo{year}{2013}\natexlab{}.
\newblock \showarticletitle{{PIE: A Lightweight Control Scheme to Address the
  Bufferbloat Problem}}. In \bibinfo{booktitle}{{\em IEEE HPSR}}.
\newblock


\bibitem[\protect\citeauthoryear{{Pat Bosshart and Glen Gibb and Hun-Seok Kim
  and George Varghese and Nick McKeown and Martin Izzard and Fernando Mujica
  and Mark Horowitz}}{{Pat Bosshart and Glen Gibb and Hun-Seok Kim and George
  Varghese and Nick McKeown and Martin Izzard and Fernando Mujica and Mark
  Horowitz}}{2013}]%
        {rmt}
\bibfield{author}{\bibinfo{person}{{Pat Bosshart and Glen Gibb and Hun-Seok Kim
  and George Varghese and Nick McKeown and Martin Izzard and Fernando Mujica
  and Mark Horowitz}}.} \bibinfo{year}{2013}\natexlab{}.
\newblock \showarticletitle{{Forwarding Metamorphosis: Fast Programmable
  Match-action Processing in Hardware for SDN}}. In \bibinfo{booktitle}{{\em
  {ACM SIGCOMM}}}.
\newblock


\bibitem[\protect\citeauthoryear{Perry, Ousterhout, Balakrishnan, Shah, and
  Fugal}{Perry et~al\mbox{.}}{2014}]%
        {fastpass}
\bibfield{author}{\bibinfo{person}{Jonathan Perry}, \bibinfo{person}{Amy
  Ousterhout}, \bibinfo{person}{Hari Balakrishnan}, \bibinfo{person}{Devavrat
  Shah}, {and} \bibinfo{person}{Hans Fugal}.} \bibinfo{year}{2014}\natexlab{}.
\newblock \showarticletitle{{Fastpass: A Centralized ``Zero-Queue" Datacenter
  Network}}. In \bibinfo{booktitle}{{\em ACM SIGCOMM}}.
\newblock


\bibitem[\protect\citeauthoryear{Rhee and Xu}{Rhee and Xu}{2008}]%
        {cubic}
\bibfield{author}{\bibinfo{person}{Injong Rhee} {and} \bibinfo{person}{Lisong
  Xu}.} \bibinfo{year}{2008}\natexlab{}.
\newblock \showarticletitle{{CUBIC: A New TCP-Friendly High-Speed TCP
  Variant}}.
\newblock \bibinfo{journal}{{\em ACM SIGOPS Operating System Review\/}}
  \bibinfo{volume}{42}, \bibinfo{number}{5} (\bibinfo{date}{July}
  \bibinfo{year}{2008}).
\newblock


\bibitem[\protect\citeauthoryear{Shan and Ren}{Shan and Ren}{2017}]%
        {cedm}
\bibfield{author}{\bibinfo{person}{Danfeng Shan} {and}
  \bibinfo{person}{Fengyuan Ren}.} \bibinfo{year}{2017}\natexlab{}.
\newblock \showarticletitle{{Improving ECN Marking Scheme with Micro-burst
  Traffic in DCN}}. In \bibinfo{booktitle}{{\em IEEE INFOCOM}}.
\newblock


\bibitem[\protect\citeauthoryear{Vamanan, Hasan, and Vijaykumar}{Vamanan
  et~al\mbox{.}}{2012}]%
        {d2tcp}
\bibfield{author}{\bibinfo{person}{Balajee Vamanan}, \bibinfo{person}{Jahangir
  Hasan}, {and} \bibinfo{person}{T.~N. Vijaykumar}.}
  \bibinfo{year}{2012}\natexlab{}.
\newblock \showarticletitle{{Deadline-Aware Datacenter TCP (D2TCP)}}. In
  \bibinfo{booktitle}{{\em ACM SIGCOMM}}.
\newblock


\bibitem[\protect\citeauthoryear{Vanini, Pan, Alizadeh, Taheri, and
  Edsall}{Vanini et~al\mbox{.}}{2017}]%
        {letflow}
\bibfield{author}{\bibinfo{person}{Erico Vanini}, \bibinfo{person}{Rong Pan},
  \bibinfo{person}{Mohammad Alizadeh}, \bibinfo{person}{Parvin Taheri}, {and}
  \bibinfo{person}{Tom Edsall}.} \bibinfo{year}{2017}\natexlab{}.
\newblock \showarticletitle{{Let It Flow: Resilient Asymmetric Load Balancing
  with Flowlet Switching}}. In \bibinfo{booktitle}{{\em 14th USENIX NSDI}}.
\newblock


\bibitem[\protect\citeauthoryear{Vishwanath, Sivaraman, and Thottan}{Vishwanath
  et~al\mbox{.}}{2009}]%
        {buffersizing}
\bibfield{author}{\bibinfo{person}{Arun Vishwanath}, \bibinfo{person}{Vijay
  Sivaraman}, {and} \bibinfo{person}{Marina Thottan}.}
  \bibinfo{year}{2009}\natexlab{}.
\newblock \showarticletitle{Perspectives on Router Buffer Sizing: Recent
  Results and Open Problems}.
\newblock \bibinfo{journal}{{\em SIGCOMM Comput. Commun. Rev.\/}}
  \bibinfo{volume}{39}, \bibinfo{number}{2} (\bibinfo{year}{2009}).
\newblock


\bibitem[\protect\citeauthoryear{Wilson, Ballani, Karagiannis, and
  Rowstron}{Wilson et~al\mbox{.}}{2011}]%
        {d3}
\bibfield{author}{\bibinfo{person}{Christo Wilson}, \bibinfo{person}{Hitesh
  Ballani}, \bibinfo{person}{Thomas Karagiannis}, {and} \bibinfo{person}{Ant
  Rowstron}.} \bibinfo{year}{2011}\natexlab{}.
\newblock \showarticletitle{{Better Never than Late: Meeting Deadlines in
  Datacenter Networks}}. In \bibinfo{booktitle}{{\em ACM SIGCOMM}}.
\newblock


\bibitem[\protect\citeauthoryear{Zhu, Era, Firestone, Guo, Lipshteyn, Liron,
  Padhye, Raindel, Yahia, and Zhang}{Zhu et~al\mbox{.}}{2015}]%
        {dcqcn}
\bibfield{author}{\bibinfo{person}{Yibo Zhu}, \bibinfo{person}{Haggai Era},
  \bibinfo{person}{Daniel Firestone}, \bibinfo{person}{Chuanxiong Guo},
  \bibinfo{person}{Marina Lipshteyn}, \bibinfo{person}{Yehonatan Liron},
  \bibinfo{person}{Jitendra Padhye}, \bibinfo{person}{Shachar Raindel},
  \bibinfo{person}{Mohamad~Haj Yahia}, {and} \bibinfo{person}{Ming Zhang}.}
  \bibinfo{year}{2015}\natexlab{}.
\newblock \showarticletitle{{Congestion Control for Large-Scale RDMA
  Deployments}}. In \bibinfo{booktitle}{{\em ACM SIGCOMM}}.
\newblock


\bibitem[\protect\citeauthoryear{Zhuo, Zhang, Liu, Krishnamurthy, and
  Anderson}{Zhuo et~al\mbox{.}}{2016}]%
        {rackcc}
\bibfield{author}{\bibinfo{person}{Danyang Zhuo}, \bibinfo{person}{Qiao Zhang},
  \bibinfo{person}{Vincent Liu}, \bibinfo{person}{Arvind Krishnamurthy}, {and}
  \bibinfo{person}{Thomas Anderson}.} \bibinfo{year}{2016}\natexlab{}.
\newblock \showarticletitle{{RackCC: Rack-level Congestion Control}}. In
  \bibinfo{booktitle}{{\em ACM HotNets-XV}}.
\newblock


\end{thebibliography}
